\newif\ifisFullVersion

\isFullVersiontrue

\ifisFullVersion
\documentclass[letter, 10pt]{article}
\newcommand{\abstractOnly}[1]{}
\newcommand{\fullOnly}[1]{#1}
\else 
\documentclass[letter, 11pt]{article}
\newcommand{\abstractOnly}[1]{#1}
\newcommand{\fullOnly}[1]{}
\fi

\usepackage{amsmath,amssymb,amsthm}
\usepackage{stmaryrd,ulsy}
\usepackage{graphics}

\usepackage{fullpage}

\newcommand{\intvl}[2]{\ensuremath{\left\llbracket #1\,;\,#2 \right\rrbracket}}
\newcommand{\preput}[1]{\mbox{#1 }\quad}
\newcommand{\subprop}[1]{\preput{{\bf #1.}}}
\newcommand{\where}[0]{\preput{where}}
\newcommand{\wehave}[0]{. We have}
\newcommand{\style}[1]{#1}
\newcommand{\textgadget}[1]{\text{{\it #1}}}
\newcommand{\litp}[1]{p_{#1}}
\newcommand{\litn}[1]{n_{#1}}

\newcommand{\ident}[2]{{\mathcal I}_{#2}^{#1}}
\newcommand{\dock}[2]{\textgadget{Dock}(#1,#2)}
\newcommand{\lock}[1]{\textgadget{Lock}(#1)}
\newcommand{\hook}[1]{\textgadget{Hook}(#1)}
\newcommand{\fork}[1]{\textgadget{Fork}(#1)}
\newcommand{\literals}[2]{\textgadget{Literals}(#1,#2)}
\newcommand{\clause}[4]{\textgadget{Clause}(#1,#2,#3,#4)}
\newcommand{\var}[3]{\textgadget{Variable}(#1,#2,#3)}
\newcommand{\key}[0]{\mbox{\it \style{key}}}
\newcommand{\test}[0]{\mbox{\it \style{test}}}
\newcommand{\take}[0]{\mbox{\it\style{take}}}
\newcommand{\puth}[0]{\mbox{\it\style{put}}}
\newcommand{\biglambda}[3]{\Lambda^{#2}_{#3}}
\newcommand{\bigP}[0]{\ensuremath{\mathcal P}}
\newcommand{\SP}[1]{\ensuremath{S_{\phi}[#1]}}
\newcommand{\SPs}[2]{\ensuremath{S_{\phi}[#1,#2]}}
\newcommand{\sel}[0]{\ensuremath{\sigma}}

\newtheorem{defn}{Definition}
\newtheorem{thm}{Theorem}
\newtheorem{property}[thm]{Property}

\newtheorem{lemma}[thm]{Lemma}

\newcommand{\botdead}{\bot}
\newcommand{\flip}{\rightarrow}
\newcommand{\dead}{\flip\botdead}
\newcommand{\path}{\Longrightarrow}
\newcommand{\pathDead}{\path\emptyset}
\newcommand{\rp}[2]{%
\hspace{1em}\makebox[0em]{\raisebox{-0.6em}{%
{\scriptsize \ref{prop:#1}.#2}%
}}\hspace{-1em}}

\newcommand{\n}{,\,}
\newcommand{\stacknext}{\\}
\newcommand{\fh}[1]{\underline{\ #1\ }}
\newcommand{\fp}[1]{\underline{#1}}
\newcommand{\fpstack}[1]{\underline{\ #1\ }}
\newcommand{\row}[1]{\ensuremath{\big< #1 \big>}}

\newcommand{\stack}[1]{{\def\n{\stacknext} \def\fp{\fpstack}  \begin{matrix}#1 \end{matrix}}}
\newcommand{\fstack}[1]{\boxed{\stack{#1}}}
\newcommand{\stackrow}[1]{\ifisFullVersion
\stack{#1} 
\else
\row{#1} 
\fi }
\newcommand{\fstackrow}[1]{\ifisFullVersion
\fstack{#1} 
\else
\row{#1} 
\fi }
\newcommand{\ch}[1]{{\def\style{\textbf} \boldsymbol {#1}}}

\newcommand{\doubleM}[2]{
\begin{array}{l}%
\raisebox{-1em}{$\nearrow$} #1\\
\\%
\raisebox{1em}{$\searrow$} #2
\end{array}
}

\newcommand{\doubleD}[2]{
\begin{array}{l}%
\phantom{\flip} #1\\
\nearrow\\%
\flip #2
\end{array}
}

\newcommand{\rev}[1]{ \raisebox{0.4em}{\footnotesize $\star$} #1}

\newcommand{\revfh}[1]{ \fh{\rev{#1}}}

\title{Pancake Flipping Is Hard \abstractOnly{\\\medskip  \small \sc Extended Abstract}}
\author{Laurent Bulteau, Guillaume Fertin, Irena Rusu\smallskip\\
\small
Laboratoire d'Informatique de Nantes-Atlantique (LINA), UMR CNRS
  6241\\ \small
  Universit\'e de Nantes, 2 rue de la Houssini\`ere, 44322 Nantes 
  Cedex 3 - France \\ \small
\texttt{ \{Laurent.Bulteau, Guillaume.Fertin, Irena.Rusu\}@univ-nantes.fr}
}
\date{}

\begin{document}

\maketitle

\begin{center}\begin{minipage}{0.80\textwidth}

{\bf Abstract.}  
Pancake Flipping is the problem of sorting a stack of pancakes of different sizes (that is, a permutation), when the only allowed operation  is to insert a spatula anywhere in the stack and to flip the pancakes above it (that is, to perform a prefix reversal). In the burnt variant, one side of each pancake is marked as burnt, and it is required to finish with all pancakes having the burnt side down. Computing the optimal scenario for any stack of pancakes and determining the worst-case stack for any stack size have been challenges over more than three decades. Beyond being an intriguing combinatorial problem in itself, it also yields applications, e.g. in parallel computing and computational biology.

In this paper, we show that the Pancake Flipping problem, in its original (unburnt) variant, is \textsf{NP}-hard, thus answering the long-standing question of its computational complexity.

\vspace{2em}


\vspace{2em}
\end{minipage}
\end{center}

\clearpage
\section{Introduction}
The pancake problem was stated in~\cite{D75} as follows:

\begin{center}
 \begin{minipage}{0.85\textwidth}\small 
The chef in our place is sloppy, and when he prepares a stack of pancakes they come out all different sizes. Therefore, when I deliver them to a customer, on the way to the table I rearrange them (so that the smallest winds up on top, and so on, down to the largest at the bottom) by grabbing several from the top and flipping them over, repeating this (varying the number I flip) as many times as necessary. If there are $n$ pancakes, what is the maximum number of flips (as a function of $n$) that I will ever have to use to rearrange them?
\end{minipage}
\end{center}

Stacks of pancakes are represented by permutations, and a flip consists in reversing a prefix of any length. The previous puzzle yields two entangled problems:
\begin{itemize}                                                                                                                                  
\item Designing an algorithm that sorts any permutation with a minimum number of flips (this optimization problem is called MIN-SBPR, for Sorting By Prefix Reversals).
\item Computing $f(n)$, the maximum number of flips required to sort a permutation of size $n$ (the diameter of the so-called \emph{pancake network}).                                                                                                                            
\end{itemize}

Gates and Papadimitriou \cite{GP79} introduced the \emph{burnt} variant of the problem: the pancakes are two-sided, and an additional constraint requires the pancakes to end with the unburnt side up. The diameter of the corresponding \emph{burnt pancake network} is denoted $g(n)$.
 A number of studies~\cite{C09,C11,CB95,GP79,HS93,HS97,LC11} have aimed at determining more precisely the values of $f(n)$ and $g(n)$, with the following results:
\begin{itemize}
 \item $f(n)$ and $g(n)$ are known exactly for $n\leq 19$ and $n\leq 17$, respectively~\cite{C11}.
 \item $15n/14 \leq f(n) \leq 18n/11+O(1) $~\cite{HS97, C09}.
 \item $\lfloor (3n+3)/2 \rfloor \leq g(n) \leq 2n-6$~\cite{C11} (upper bound for $n\geq 16$).
\end{itemize}
Considering MIN-SBPR, 2-approximation algorithms have been designed, both for the burnt~\cite{CB95,FG05} and unburnt~\cite{FG05} variants. Moreover, Labarre and Cibulka~\cite{LC11} have characterized a subclass of permutations, called \emph{simple permutations}, that can be sorted in polynomial time.

The pancake problems have various applications. For instance, the pancake network, having both a small degree and diameter, is of interest in parallel computing. The algorithmic aspect, i.e. the sorting problem, has applications in comparative genomics, since prefix reversals are possible elementary modifications that can affect a genome during evolution. A related problem is Sorting By Reversals~\cite{BP93} where any subsequence can be flipped at any step, not only prefixes. This problem is now well-known, with a polynomial-time exact algorithm~\cite{HP95} for the signed case, and a 1.375-approximation~\cite{BHK02} for the \textsf{APX}-hard unsigned case~\cite{BK99}.

In this paper, we prove that the MIN-SBPR problem is \textsf{NP}-hard (in its unburnt variant), thus answering an open question raised several decades ago. We in fact prove a stronger result: it is known that the number of breakpoints of a permutation (that is, the number of pairs of consecutive elements that are not consecutive in the identity permutation) is a lower bound on the number of flips necessary to sort a permutation. We show that deciding whether this bound is tight is already \textsf{NP}-hard.

\section{Notations}
We denote by \intvl{a}{b} the interval $\{a,a+1,\ldots,b\}$ (for $b<a$, we have $\intvl ab=\emptyset$).
Let $n$ be an integer. Input sequences are permutations of \intvl1n, hence we consider only sequences where all elements are unsigned, and there cannot be duplicates.
\fullOnly{When there is no ambiguity, we use the same notation for a sequence and the set of elements it contains.}
We use upper case for \fullOnly{sets and} sequences, and lower case for elements.

Consider a sequence $S$ of length $n$, $S=\row{ x_1\n x_2\n \ldots\n x_n}$. Element $x_1$ is said to be the \textit{head element} of $S$. Sequence $S$ has a \textit{breakpoint} at position $r$, $1\leq r < n$ if $x_r\neq x_{r+1}-1$ and $x_r\neq x_{r+1}+1$. It has a \textit{breakpoint} at position $n$ if $x_n \neq n$. We write $d_b(S)$ the number of breakpoints of $S$. Note that having $x_1\neq 1$ does not directly count as a breakpoint, and that $d_b(S)\leq n$ for any sequence of length $n$. 
For any $p \leq q\in \mathbb N$, we write $\ident{p}{q}$ the sequence $\row{p\n p+1\n p+2\n\ldots\n q}$. $\ident{1}{n}$ is the \textit{identity}.
For a sequence of any length $S= \row{x_1\n x_2\n \ldots\n x_k}$, we write $\rev S$ the sequence obtained by reversing $S$: $\rev S=\row{x_k\n x_{k-1}\n \ldots\n x_1}$. Given an integer $p$, we write $p+S= \row{p+x_1\n p+x_2\n \ldots \n p+ x_k}$.

The \textit{flip} of length $r$ is the operation that consists in reversing the $r$ first elements of the sequence. It transforms 
\begin{equation*}
 S= \row{x_1\n x_2\n \ldots \n x_r\n x_{r+1}\n \ldots \n x_n}
\end{equation*} into 
\begin{equation*}
 S'=\row{x_r\n x_{r-1}\n \ldots \n x_1\n x_{r+1} \n\ldots \n x_n}.
\end{equation*}
Note that the flip of length $1$ does not modify $S$, and the flip of length $n$ transforms $S$ into $\rev S$. 

\begin{property}
Given a sequence $S'$ obtained from a sequence $S$ by performing one flip,
we have $d_b(S')-d_b(S)\in \{-1,0,1\}$.
\end{property}

A flip from $S$ to $S'$ is said to be \textit{efficient} if $d_b(S')=d_b(S)-1$, and we reserve the notation $S\flip S'$ for such flips. A sequence of size $n$, different from the identity, is a \textit{deadlock} if it yields no efficient flip, and we write $S\dead$. 
By convention, we underline in a sequence the positions corresponding to possible efficient flips: there are at most two of them, and at least one if the sequence is neither a deadlock nor the identity.

We call \textit{path} a series of flips. 
A path is \textit{efficient} if each flip is efficient in the series. A sequence $S$ is \textit{efficiently sortable} if there exists an efficient path from $S$ to the identity permutation (equivalently, if it can be sorted in $d_b(S)$ flips). See for example Figure~\ref{fig:ex1}.

\begin{figure}
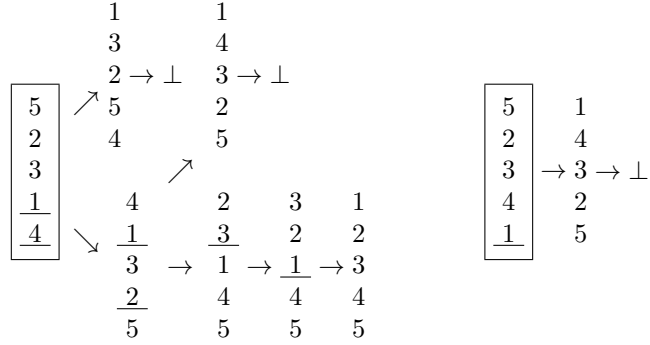

 \ifisFullVersion
\begin{equation*}
 \boxed{\stack{5\n 2\n 3\n \fp{1}\n \fp{4}}}
 \doubleM{\ \stack{1\n 3\n 2\n 5\n 4} \dead }{
 \ \  \stack{4\n \fp{1}\n 3\n \fp{2}\n 5}}%
\hspace{-1.6em}
\doubleD{\,\ \  \stack{1\n 4\n 3\n 2 \n 5}\dead }{
\; \stack{2\n \fp{3}\n 1\n 4\n 5}
\flip \stack{3\n 2\n \fp{1} \n 4\n 5}
\flip \stack{1\n 2\n 3 \n 4\n 5}  }
\qquad\qquad
 \boxed{\stack{5\n 2\n 3\n 4\n \fp{1}}}
\flip \stack{1\n 4 \n 3\n 2\n 5}\dead 
\end{equation*} 
\fi
\caption{\label{fig:ex1}Examples of efficient flips. Sequence $\row{5\n 2\n 3\n 1\n 4}$ is efficiently sortable (in four flips), but $\row{5\n 2\n 3\n 4\n 1}$ is not.}
\end{figure}

Let $S$ be a sequence different from the identity, and $\mathbb T$ be a set of sequences. We write $S\path \mathbb T$ if both following conditions are satisfied:
\begin{enumerate}
\item for each $T\in \mathbb T$, there exists an efficient path from $S$ to $T$.
\item for each efficient path from $S$ to the identity, there exists a sequence $T\in \mathbb T$ such that the path goes through $T$.                                                                                                            
\end{enumerate}
If $\mathbb T$ consists of a single element ($\mathbb T=\{T\}$), we may write $S\path T$ instead of $S\path \{T\}$.
Note that condition~1. is trivial if $\mathbb T=\emptyset$, and condition~2. is trivial if there is no efficient path from $S$ to $\ident 1n$.
Note that given a sequence $S$, there can be several different sets $\mathbb T$ such that $S\path \mathbb T$. %
\fullOnly{However, two are especially relevant:}%
\abstractOnly{

The following properties are easily deduced from the definition of $\path$.
}
\begin{property}Given any sequence $S\neq \ident 1n$,
\begin{eqnarray*}
 S\path \ident1n&\Leftrightarrow& S\text{ is efficiently sortable.}\\
 S\path \emptyset&\Leftrightarrow& S \text{ is not efficiently sortable.}
\end{eqnarray*}

\end{property}
\ifisFullVersion\begin{proof}
For $S\path \ident1n$: condition 1. is true iff there exists an efficient path from $S$ to the identity, that is $S$ is efficiently sortable. Condition 2. is always true.

For $S\path \emptyset$: condition 1. is always true. If there exists at least one efficient path from $S$ to $\ident 1n$, then, since there exists no sequence $T\in \emptyset$, Condition 2. cannot be true. Hence Condition 2. is false when there exists an efficient path from $S$ to the identity and true otherwise, so it is equivalent to the fact that $S$ is not efficiently sortable.
\end{proof}\fi

\fullOnly{The following property is easily deduced from the definition.}
\begin{property}
If $S\path \{S_1, S_2\}$,  $S_1\path \mathbb T_1$ and $S_2\path \mathbb T_2$, then $S\path \mathbb T_1 \cup  \mathbb T_2$.
\end{property}

\section{Reduction from 3-SAT}
The reduction uses a number of gadget sequences in order to simulate boolean variables and clauses with subsequences. They are organized in two levels (where level-1 gadgets are directly defined by sequences of integers, and level-2 gadgets are defined using a pattern of level-1 gadgets). For each gadget we define, we derive a property characterizing the efficient paths that can be followed if some part of the gadget appears at the head of a sequence. 
\abstractOnly{The proofs for all these properties follow the same pattern, with no obstacle appart from the increasing complexity of the sequences, and only the one for the Dock gadget is given in this extended abstract.}

\ifisFullVersion
We have not aimed at providing the smallest possible gadgets (the overall reduction for a formula containing $l$ variables and $k$ clauses creates a stack of $31l+98k$ elements with $16l+50k$ breakpoints), and we preferred straightforward proofs and easy-to-combine gadgets over short sequences. A rough analysis shows that the final stack size could easily be reduced to $22l+71k$, with the same number of breakpoints.
\fi

\subsection{Level-1 gadgets}
\subsubsection{Docks}

The dock gadget is the simplest we define. Its only goal is to store sequences of the kind $\rev\ident{p+1}{q}$ 
(
with $p<q$) out of the head of the sequence, without ``disturbing'' any other part.
\begin{defn} Given two integers $p$ and $q$ with $p<q$, the \emph{dock} for $\rev\ident{p+1}{q}$ is the sequence
\begin{eqnarray*}
\dock{p}{q}&=& D  \\
\where
 D&=&\row{ p-1 \n p\n q+1\n q+2}.
\end{eqnarray*}
\end{defn}

It has the following property:
\begin{property}\label{prop:dock} 
Let $p$ and $q$ be any integers with $p<q$, $D=\dock{p}{q}$, and $X$ and $Y$ be any sequences\wehave 
\begin{equation*}
\stackrow{\rev{\ident{p+1}{q}} \n X \n D \n  Y }
 \path 
\stackrow{X \n \ident{p-1}{q+2} \n  Y } 
\end{equation*}
\end{property}

\ifisFullVersion
\begin{proof}
 An efficient path from \row{\rev{\ident{p+1}{q}} \n X \n D \n  Y } to \row{X \n \ident{p-1}{q+2} \n  Y } is given in Figure~\ref{fig:proof_dock}. For each sequence in the path, we apply the only possible efficient flip, hence every efficient path between \row{\rev{\ident{p+1}{q}} \n X \n D \n  Y } and $\ident 1n$ (if such a path exists) begins with these two flips, and goes through \row{X \n \ident{p-1}{q+2} \n  Y }.
\end{proof}

\begin{figure}
\begin{equation*}
\boxed{ \stack{\rev{\ident{p+1}{q}} \n X \n D \n  Y }}
=
\stack{ q\n q-1\n \vdots \n p+2\n p+1 \n X \n p-1 \n \fh{p}\n q+1\n q+2 \n Y}
\flip
\stack{  p \n p-1\n \revfh{ X} \n p+1\n p+2\n \vdots \n q-1\n q\n q+1\n q+2 \n Y}
\flip
\stack{ X\n p-1 \n p \n p+1\n p+2\n \vdots \n q-1\n q\n q+1\n q+2 \n Y}
=
\boxed{ \stack{X \n \ident{p-1}{q+2} \n  Y } }
\end{equation*}
\caption{\label{fig:proof_dock}Proof of Property~\ref{prop:dock}. (Dock gadget)}
\end{figure}

\fi

\subsubsection{Lock}
A lock gadget contains three parts: a sequence which is the lock itself, a key element that ``opens'' the lock, and a test element that checks whether the lock is open. 

\begin{defn} For any integer $p$, $\lock{p}$ is defined by
\begin{eqnarray*}
\lock{p} &=& (\key,\test,L) \\
\where
\key&=&p+10\\
\test&=&p+7\\
L&=& p+\row{1\n 2 \n 9 \n 8 \n 5 \n 6 \n 4 \n 3 \n 11 \n 12}\\
\end{eqnarray*}

Given a lock $ (\key,\test,L) = \lock{p} $, we write 
\begin{equation*}
L^o =  p+\row{1\n 2 \n 3 \n 4 \n 6 \n 5 \n 8 \n 9 \n 10 \n 11 \n 12} . 
\end{equation*}
\end{defn}

 Sequences $L$ and $L^o$ represent the lock when it is respectively closed or open.
If a sequence containing a closed lock has $\key$ for head element, then efficient flips put the lock in open position. 
If it has $\test$ for head element, then it is a deadlock if and only if the lock is closed.
\begin{property}\label{prop:lock}

Let $p$ be any integer, $ (\key,\test,L) = \lock{p} $, and $X$ and $Y$ be any sequences\wehave
\ifisFullVersion
\begin{equation*}
\subprop{a} 
\stack{\key\n X \n L \n Y} \path \stack{X \n L^o \n Y}
\qquad\qquad
\subprop{b}
 \stack{\test\n X \n L^o \n Y} \path \stack{X \n \ident{p+1}{p+12} \n Y}
\qquad \qquad
\subprop{c}
 \stack{\test\n X \n L \n Y}\dead
\end{equation*}
\fi
\end{property}

\ifisFullVersion\begin{proof}
 See Figure~\ref{fig:proof_lock}. Note that for readability reasons, the proof is given for $p=0$. It can obviously be extended to any value of $p$ (each element would then be increased by $p$). 
\end{proof}
\begin{figure}
\begin{equation*}
\subprop{a} 
\fstack{\key\n X \n L \n Y}
=
\stack{10 \n X \n 1\n \fh{2} \n 9 \n 8 \n 5 \n 6 \n 4 \n \fh{3} \n 11 \n 12 \n Y}
\doubleM{
\stack{2\n 1\n \rev X \n 10 \n 9 \n 8 \n 5 \n 6 \n 4 \n 3 \n 11 \n 12 \n Y}
\dead
}{
\stack{3\n 4\n 6\n 5\n 8\n \fh{9}\n 2\n 1\n \rev X \n 10 \n 11 \n 12 \n Y}
\flip
\stack{9\n 8\n 5\n 6\n 4\n  3 \n 2\n 1\n \revfh X \n 10 \n 11 \n 12 \n Y}
\flip
\stack{X\n 1\n 2\n 3\n 4\n 6\n 5\n 8\n 9\n 10 \n 11 \n 12 \n Y}
=
\fstack{X \n L^o \n Y}
}
\qquad
\subprop{b} 
\fstack{\test\n X \n L^o \n Y}
=
\stack{7 \n X\n 1\n 2\n 3\n \fh4\n 6\n \fh5\n 8\n 9\n 10 \n 11 \n 12 \n Y}
\doubleM{
\stack{4 \n 3\n 2\n 1\n \rev X\n 7 \n 6\n 5\n 8\n 9\n 10 \n 11 \n 12 \n Y}
\dead
}{
\stack{5\n \fh6\n 4 \n 3\n 2\n 1\n \rev X\n 7 \n 8\n 9\n 10 \n 11 \n 12 \n Y}
\flip
\stack{6\n 5\n 4 \n 3\n 2\n 1\n \revfh X\n 7 \n 8\n 9\n 10 \n 11 \n 12 \n Y}
\flip
\stack{X\n 1\n 2\n 3\n 4\n 5\n 6\n 7\n 8\n 9 \n10 \n 11 \n 12 \n Y}
=
\fstack{X \n \ident{1}{12} \n Y}
}
\end{equation*}
\begin{equation*}
\subprop{c} 
\fstack{\test\n X \n L \n Y}
=
\stack{7 \n X \n 1\n 2 \n 9 \n 8 \n 5 \n 6 \n 4 \n 3 \n 11 \n 12 \n Y}
\dead
\end{equation*}
\caption{\label{fig:proof_lock}Proof of Property~\ref{prop:lock}. (Lock gadget)}
\end{figure}\fi
We use locks to emulate literals of a boolean formula: variables ``hold the keys'', and  in a first time open the locks corresponding to true literals. Each clause holds three test elements, corresponding to its three literals, and the clause is true if the lock is open for at least one of the test elements.
\subsubsection{Hook}

A hook gadget contains four parts: two sequences used as delimiters, a \textit{take} element that takes the interval between the delimiters and places it in head, and a \textit{put} element that does the reverse operation. Thus, the sequence between the delimiters can be stored anywhere until it is called by $\take$, and then can be stored back using $\puth$. 
\begin{defn}\label{def:hook}
 For any integer $p$, $\hook{p}$ is defined by
\begin{eqnarray*}
\hook{p}&=&(\take, \puth, G, H) \\
\where
\take &=& p+10 \\
\puth &=& p+7 \\
G &=& p+ \row{3\n 4} \\
H&=& p+ \row{12\n 11\n 6\n 5\n 9\n 8\n 2\n 1}.
\end{eqnarray*}

Given a hook $(\take, \puth, G, H)=\hook{p} $, we write
\begin{eqnarray*}
G' &=& p+ \row{12\n 11\n 6\n 5\n 4\n 3} \\
H' &=& p+ \row{10\n 9\n 8\n 2\n 1} \\
G'' &=& p+ \row{3\n 4\n 5\n 6\n 7} \\
H'' &=& p+ \row{12\n 11\n 10\n 9\n 8\n 2\n 1}.
\end{eqnarray*} 
\end{defn}

\begin{property}\label{prop:hook}
Let $p$ be an integer,  $(\take, \puth, G, H)=\hook{p} $, and  $X$, $Y$ and $Z$ be any sequences\wehave
\ifisFullVersion
\begin{equation*}
\subprop a
\stack{\take \n X \n G \n Y \n H\n Z} \path \stack{Y \n G'\n \rev X \n H' \n Z}
\qquad\qquad
\subprop b
 \stack{\puth \n X \n G'\n \rev Y \n H' \n Z} \path \stack{Y \n G''\n X \n H'' \n Z}
\qquad\qquad
\subprop c
 \stack{G''\n X \n H'' \n Y} \path \stack{X \n \rev\ident{p+1}{p+12} \n Y}
\end{equation*}
\fi
\end{property}
\ifisFullVersion\begin{proof}
 See Figure~\ref{fig:proof_hook} (with $p=0$).
\end{proof}\fi
\ifisFullVersion\begin{figure}
 \begin{equation*}
\subprop{a} 
 \fstack{\take \n X \n G \n Y \n H\n Z}
=
  \stack{10 \n X \n 3\n 4 \n Y \n 12\n 11\n 6\n \fh{5}\n 9\n 8\n 2\n 1\n Z}
\flip
  \stack{5 \n 6 \n 11 \n 12 \n \revfh Y \n 4 \n 3 \n \rev X \n 10 \n 9\n 8\n 2\n 1\n Z}
\flip
  \stack{Y \n 12\n 11\n 6\n 5\n 4 \n 3 \n \rev X \n 10 \n 9\n 8\n 2\n 1\n Z}
=
 \fstack{Y \n  G' \n \rev X \n H'\n Z}%
 \end{equation*}
 \begin{equation*}
\subprop{b} 
 \fstack{\puth \n X \n G'\n \rev Y \n H' \n Z}
=
\stack{7 \n X \n 12\n \fh{11}\n 6\n 5\n 4 \n 3 \n \rev Y \n 10 \n 9\n 8\n 2\n 1\n Z}
\flip
\stack{11 \n 12  \n \rev X  \n 7 \n 6\n 5\n 4 \n 3 \n \revfh Y \n 10 \n 9\n 8\n 2\n 1\n Z}
\flip
\stack{ Y \n 3 \n 4 \n 5 \n 6 \n 7 \n X \n 12 \n 11 \n 10 \n 9\n 8\n 2\n 1\n Z}
=
 \fstack{Y \n G''\n X \n H'' \n Z}
\qquad
\subprop c
 \fstack{G''\n X \n H'' \n Y} 
=
\stack{ 3 \n 4 \n 5 \n 6 \n 7 \n X \n 12 \n 11 \n 10 \n 9\n \fh 8\n 2\n 1\n Y}
\flip
\stack{8 \n 9 \n 10 \n 11 \n 12 \n \revfh X \n 7 \n 6 \n 5 \n 4 \n 3 \n 2 \n 1 \n Y}
\flip
\stack{X \n 12 \n 11 \n 10 \n 9\n 8\n 7 \n 6 \n 5 \n 4 \n 3\n 2\n 1\n Y}
=
 \fstack{X \n \rev\ident{1}{12} \n Y}
 \end{equation*}
\caption{\label{fig:proof_hook}Proof of Property~\ref{prop:hook}. (Hook Gadget)}
\end{figure}\fi

\subsubsection{Fork}

A fork gadget implements choices. It contains two parts delimiting a sequence $X$. Any efficient path encountering a fork gadget follows one of two tracks, where either $X$ or $\rev X$ appears at the head of the sequence at some point. Sequence $X$ would typically contain a series of triggers for various gadgets ($\key$, $\take$, etc.), so that $X$ and $\rev X$ differ in the order in which the gadgets are triggered.

\begin{defn}\label{def:fork}
 For any integer $p$, $\fork{p}$ is defined by
\begin{eqnarray*}
\fork{p}&=&(E, F) \\
\where
E &=& p+ \row{11\n 8\n 7\n 3} \\
F &=& p+ \row{10\n 9\n 6\n 12\n 13\n 4\n 5\n 15\n 14\n 2\n 1}.
\end{eqnarray*}

Given a fork $(E,F)=\fork{p} $, we write
\begin{eqnarray*}
F^1 &=& p+ \row{10\n 9\n 6\n 7\n 8\n 11\n 12\n 13\n 14\n 15\n 5\n 4\n 3\n 2\n 1} \\
F^2 &=& p+ \row{3\n 7\n 8\n 11\n 10\n 9\n 6\n 12\n 13\n 4\n 5\n 15\n 14\n 2\n 1} \\
\end{eqnarray*} 
\end{defn}

\begin{property}\label{prop:fork}
Let $p$ be an integer,  $(E,F)=\fork{p} $, and  $X$, $Y$ be any sequences\wehave
\ifisFullVersion
\begin{equation*}
\subprop a
 \stack{E\n X\n F \n Y} \path \left\{ \stack{X \n  F^1 \n Y}\ ,\   \stack{\rev X \n  F^2 \n Y} \right\}
\qquad\qquad
\subprop b
\stack{F^1 \n Y} \path \stack{\rev \ident{p+1}{p+15} \n Y}
\qquad\qquad
\subprop c
\stack{F^2 \n Y} \path \stack{\rev \ident{p+1}{p+15} \n Y}
\end{equation*}
\fi
\end{property}

\ifisFullVersion\begin{proof}
 See Figures~\ref{fig:proof_fork} and~\ref{fig:proof_forkb} (with $p=0$).
\end{proof}\fi
\ifisFullVersion\begin{figure}
 \begin{equation*}
\subprop{a} 
 \fstack{E\n X\n F \n Y} 
=
\stack{11\n 8\n 7\n 3 \n \fh{X}\n 10\n 9\n \fh{6}\n 12\n 13\n 4\n 5\n 15\n 14\n 2\n 1 \n Y} 
\doubleM{
\stack{\rev X \n 3\n 7\n 8\n 11\n 10\n 9\n 6\n 12\n 13\n 4\n 5\n 15\n 14\n 2\n 1 \n Y} 
=
\fstack{\rev X \n  F^2 \n Y}
}{
\ 
\stack{6\n 9\n 10\n \rev X \n \fh{3}\n 7\n 8\n 11\n 12\n 13\n 4\n 5\n 15\n 14\n 2\n 1 \n Y}
\  
\flip
\ 
\stack{3 \n X\n 10\n 9\n 6\n 7\n 8\n 11\n 12\n \fh{13}\n 4\n 5\n 15\n \fh{14}\n 2\n 1 \n Y} 
}
\doubleD{
\stack{13\n 12\n 11\n 8\n 7\n 6\n 9\n 10\n \rev X\n 3\n 4\n 5\n 15\n 14\n 2\n 1 \n Y} 
\dead
}{
\stack{14\n 15\n 5\n \fh{4}\n 13\n 12\n 11\n 8\n 7\n 6\n 9\n 10\n \rev X\n 3\n 2\n 1 \n Y} 
\flip 
\stack{4\n 5\n 15\n 14\n 13\n 12\n 11\n 8\n 7\n 6\n 9\n 10\n \revfh X\n 3\n 2\n 1 \n Y} 
\flip 
\stack{X\n 10\n 9\n 6\n 7\n 8\n 11\n 12\n 13\n 14\n 15\n 5\n 4 \n 3\n 2\n 1 \n Y} 
=
\fstack{X \n  F^1 \n Y}
}\subprop{b} 
\fstack{ F^1 \n Y}
=
\stack{ 10\n 9\n 6\n 7\n \fh8 \n 11\n 12\n 13\n 14\n 15\n 5\n 4 \n 3\n 2\n 1 \n Y} 
\flip
\stack{ 8\n 7\n \fh6\n 9\n 10\n 11\n 12\n 13\n 14\n 15\n 5\n 4 \n 3\n 2\n 1 \n Y} 
\flip
\stack{ 6\n 7\n 8\n 9\n 10\n 11\n 12\n 13\n 14\n \fh{15}\n 5\n 4 \n 3\n 2\n 1 \n Y} 
\flip
\stack{ 15\n 14\n 13\n 12\n 11\n 10\n 9\n 8\n 7\n 6\n 5\n 4 \n 3\n 2\n 1 \n Y} 
=
\fstack{\rev\ident{1}{15} \n Y}
 \end{equation*}
\caption{\label{fig:proof_fork}Proof of Properties~\ref{prop:fork}.a and \ref{prop:fork}.b (Fork gadget).}

\end{figure}\fi
\ifisFullVersion\begin{figure}
 \begin{equation*}
\subprop{c} 
\fstack{ F^2 \n Y}
=
\stack{3\n 7\n 8\n 11\n 10\n 9\n 6\n 12\n \fh{13}\n 4\n 5\n 15\n \fh{14}\n 2\n 1 \n Y} 
\doubleM{
\stack{13 \n 12\n 6\n 9\n 10\n 11\n 8\n 7\n 3\n 4\n 5\n 15\n 14\n 2\n 1 \n Y} 
\dead 
}{
\stack{14\n 15\n 5\n \fh{4}\n 13\n 12\n 6\n 9\n 10\n 11\n 8\n 7\n 3\n 2\n 1 \n Y} 
\flip
\stack{4\n 5\n 15\n 14\n 13\n 12\n 6\n 9\n 10\n 11\n 8\n \fh{7}\n 3\n 2\n 1 \n Y} 
\flip
\stack{7\n 8\n 11\n 10\n \fh{9}\n 6\n 12\n 13\n 14\n 15\n 5\n 4 \n 3\n 2\n 1 \n Y} 
\flip
\stack{9\n 10\n \fh{11}\n 8\n 7\n 6\n 12\n 13\n 14\n 15\n 5\n 4 \n 3\n 2\n 1 \n Y} 
\flip
\stack{11\n 10\n 9\n 8\n 7\n \fh{6}\n 12\n 13\n 14\n 15\n 5\n 4 \n 3\n 2\n 1 \n Y} 
\flip
\stack{6\n 7\n 8\n 9\n 10\n 11\n 12\n 13\n 14\n \fh{15}\n 5\n 4 \n 3\n 2\n 1 \n Y} 
\flip
\stack{ 15\n 14\n 13\n 12\n 11\n 10\n 9\n 8\n 7\n 6\n 5\n 4 \n 3\n 2\n 1 \n Y} 
=
\fstack{\rev\ident{1}{15} \n Y}
}
 \end{equation*}
\caption{\label{fig:proof_forkb}Proof of Property~\ref{prop:fork}.c (Fork gadget).}
\end{figure}\fi

\subsection{Level-2 gadgets}

\ifisFullVersion
In this section, we define new gadgets based on the four level-1 gadgets. From now on, each property proof uses exclusively properties from smaller gadgets. In order to help the reader follow the ever-present references, we use the following notations.
Bold font is used to emphasise the ``active'' parts of the gadget currently having an element at the head of the sequence. 
For each relation $S\path  T$, we give the relevant reference below (e.g. $S \rp{dock}{}\path  T$ if it is obtained from Property~\ref{prop:dock}).
Finally, a summary of all gadget properties (either level-1 or -2) is given in Figure~\ref{fig:compilprop}.

\begin{figure}

 
\begin{align*}
\text{\bf Dock gadget}&\\
\row{\ch{\rev{\ident{p+1}{q}}} \n X \n \ch{D} \n  Y }  &\rp{dock}{} \path \row{X \n \ident{p-1}{q+2} \n  Y } \\
\text{\bf Lock gadget}&\\
\row{\ch{\key}\n X \n \ch{L} \n Y} & \rp{lock}{a}\path \row{X \n L^o \n Y} \\
 \row{\ch{\test}\n X \n \ch{L^o} \n Y} & \rp{lock}{b}\path \row{X \n \ident{p+1}{p+12} \n Y} \\
 \row{\ch{\test}\n X \n\ch{ L} \n Y}&\rp{lock}{c}\ \dead\\
\text{\bf Hook gadget}&\\
\row{\ch{\take }\n X \n\ch{ G} \n Y \n \ch{H}\n Z} &\rp{hook}{a}\path \row{Y \n G'\n \rev X \n H' \n Z}\\
 \row{\ch{\puth }\n X \n \ch{G'}\n \rev Y \n \ch{H'} \n Z} &\rp{hook}{b}\path \row{Y \n G''\n X \n H'' \n Z}\\
 \row{\ch{G''}\n X \n \ch{H''} \n Y}& \rp{hook}{c}\path \row{X \n \rev\ident{p+1}{p+12} \n Y}\\
\text{\bf Fork gadget}&\\
\row{\ch{E}\n X\n\ch{ F} \n Y} & \rp{fork}{a}\path \left\{ 
\begin{matrix} \row{X \n  F^1 \n Y}\phantom{\rev{}} \\ \row{\rev X \n  F^2 \n Y} \end{matrix}
\right\}\\
\row{\ch{F^1} \n Y} &\rp{fork}{b} \path \row{\rev \ident{p+1}{p+15} \n Y}\\
\row{\ch{F^2} \n Y} & \rp{fork}{c}\path \row{\rev \ident{p+1}{p+15} \n Y}
\\
\text{\bf Literals gadget}&\\
\forall i\notin O\cup I,\    \row{\ch{\key_i}\n X \n \ch{\biglambda{C}{O}{I}} } 
&\rp{literals}{a}\path
 \row{X \n \biglambda{C-\{i\}}{O\cup\{i\}}{I}  }\\
  \forall i\in O,\   \row{\ch{\test_i}\n X \n \ch{\biglambda{C}{O}{I}}} 
&\rp{literals}{b}\path
 \row{X \n \biglambda{C}{O-\{i\}}{I\cup\{i\}}  }\\
 \forall i\notin O ,\   \row{\ch{\test_i}\n X \n \ch{\biglambda{C}{O}{I}} } 
&\rp{literals}{c}\ \dead
\\
\text{\bf Variable gadget}&\\
\row{\ch{\nu}\n X\n \ch{V}\n Y\n \ch{\biglambda{C}{O}{I}}} 
&\rp{variable}{a}\path
\left\{\begin{matrix}
  \row{ X\n V^1 \n Y\n \biglambda{C-P}{O\cup P}{I}}\\
  \row{ X\n V^2 \n Y\n \biglambda{C-N}{O\cup N}{I}}
\end{matrix}
\right\}\\
\row{\ch{V^1}\n X\n \ch{D} \n Y \n \ch{\biglambda{C}{O}{I}}} 
&\rp{variable}{b}\path
\row{X\n \ident{p+1}{p+31} \n Y \n \biglambda{C-N}{O\cup N}{I}} \\
\row{\ch{V^2}\n X\n \ch{D} \n Y \n \ch{\biglambda{C}{O}{I}} }
&\rp{variable}{c}\path
\row{X\n \ident{p+1}{p+31} \n Y \n \biglambda{C-P}{O\cup P}{I}} 
\\
\text{\bf Clause gadget}&\\
\row{ \ch{\gamma} \n X \n \ch{\Gamma} \n Y  \n\ch{\biglambda{C}{O}{I} }} &\rp{clause}{} \path\left\{
\begin{matrix}
\row{  X \n \Gamma^1\n Y   \n \biglambda{C}{O-\{a\}}{I\cup\{a\}} }  \mbox{ iff } a\in O
\\
\row{  X \n \Gamma^2\n Y  \n \biglambda{C}{O-\{b\}}{I\cup\{b\}} }  \mbox{ iff } b\in O
\\
\row{ X \n \Gamma^3\n Y   \n \biglambda{C}{O-\{c\}}{I\cup\{c\}} } \mbox{ iff } c\in O
\end{matrix}\right\}
\\
 \row{\ch{ \Gamma^1}\n Y\n
 \Delta\n Z\n \ch{\biglambda{C}{O}{I}} } 
&\rp{clause2}{a} \path 
\row{Y \n \ident{p+1}{p+62}\n Z\n \biglambda{C}{O-\{b,c\}}{I\cup\{b,c\}} } \\
\row{ \ch{\Gamma^2}\n Y\n
 \Delta\n Z\n \ch{\biglambda{C}{O}{I}} } 
&\rp{clause2}{b}\path 
\row{Y \n \ident{p+1}{p+62}\n Z\n \biglambda{C}{O-\{a,c\}}{I\cup\{a,c\}} } \\
 \row{ \ch{\Gamma^3}\n Y\n
 \Delta\n Z\n\ch{ \biglambda{C}{O}{I}} } 
&\rp{clause2}{c} \path 
\row{Y \n \ident{p+1}{p+62}\n Z\n \biglambda{C}{O-\{a,b\}}{I\cup\{a,b\}} } 
\end{align*} 

\caption{\label{fig:compilprop}Compilation of all gadget properties. As a general rule, $X$, $Y$, $Z$ can be any sequences, $O$ and $I$ any disjoint subsets of $\intvl 1m$. See respective definitions and properties for specific constraints and notations}
 \end{figure}\fi
\subsubsection{Literals}
The following gadget is used only once in the reduction. It contains the locks corresponding to all literals of the formula. 

\begin{defn}
Let $p$ and $m$ be two integers, $\literals pm$ is defined by
\begin{eqnarray*}
\literals{p}{m}&=&(\key_1,\ldots,\key_m,\test_1,\ldots,\test_m, \Lambda) \\
\where
\Lambda&=&\row{L_1\n L_2\n \ldots \n L_m}\\
\forall i\in\intvl1m,\ (\key_i, \test_i, L_i) & = &\lock{p+12(i-1)}
\end{eqnarray*}

Let  $O$ and $I$ be two disjoint subsets of $\intvl1m$. We write $\biglambda{C}{O}{I}$ the sequence obtained from $\Lambda$ by
\ifisFullVersion
 \begin{itemize}
\item replacing $L_i$ by $L^o_i$ for all $i\in O$,
\item replacing $L_i$ by $\ident{p+12i-11}{p+12i}$ for all $i\in I$. 
\end{itemize}
\fi

\end{defn}

Elements of $O$ correspond to open locks in $\biglambda COI$, while elements of $I$ correspond to open locks which have moreover been tested.
Note that $\biglambda{\intvl1m}{\emptyset}{\emptyset}=\Lambda$, and that $\biglambda{\emptyset}{\emptyset}{\intvl1m}=\ident{p+1}{p+12m}$.

\begin{property}\label{prop:literals}
Let $p$ and $m$ be two integers, $(\key_1,\ldots,\key_m,\test_1,\ldots,\test_m, \Lambda)=\literals{p}{m}$, $O$ and $I$ be two disjoint subsets of $\intvl1m$, and $X$ be any sequence\wehave
\begin{align*}
&\subprop a 
\forall i\in \intvl 1m-O-I,\quad   \stackrow{\key_i\n X \n \biglambda{C}{O}{I} } 
\path
 \stackrow{X \n \biglambda{C-\{i\}}{O\cup\{i\}}{I}  }
\\
&\subprop b
 \forall i\in O,\quad  \stackrow{\test_i\n X \n \biglambda{C}{O}{I}} 
\path
 \stackrow{X \n \biglambda{C}{O-\{i\}}{I\cup\{i\}}  }
\\
&\subprop c
 \forall i\in \intvl 1m- O ,\quad  \stackrow{\test_i\n X \n \biglambda{C}{O}{I} } 
\dead
\end{align*}
\end{property}
\ifisFullVersion\begin{proof}

 The proof follows from Property~\ref{prop:lock}. 

\subprop a Let $i\in \intvl 1m-O-I$. Then $\biglambda{C}{O}{I}$ can be written
$\biglambda{C}{O}{I}=\row{A \n L_i \n B}$. Hence 
\begin{align*}
& \row{\key_i\n X \n \biglambda{C}{O}{I}}
\\
&= \row{\ch{\key_i}\n X \n A \n \ch{L_i} \n B}
\\
&\rp{lock}{a}\path \row{ X \n A \n L_i^o \n B}
\\
&= \row{X \n \biglambda{C-\{i\}}{O\cup\{i\}}{I}}
\end{align*}

\subprop b Let $i\in O$. Then $\biglambda{C}{O}{I}$ can be written
$\biglambda{C}{O}{I}=\row{A \n L_i^o \n B}$. Hence 
\begin{align*}
& \row{\test_i\n X \n \biglambda{C}{O}{I}}
\\
&= \row{\ch{\test_i}\n X \n A \n \ch{L_i^o} \n B}
\\
&\rp{lock}{b}\path \row{ X \n A \n \ident{p+12i-11}{p+12i} \n B}
\\
&= \row{X \n \biglambda{C}{O-\{i\}}{I\cup\{i\}}}
\end{align*}

\subprop c Let $i\in \intvl 1m - O$. If $i\in I$, then $\test_i\in \ident{p+12i-11}{p+12i}\subset \biglambda{C}{O}{I}$, and \row{\test_i\n X \n \biglambda{C}{O}{I}} is not a valid sequence (it contains a duplicate). Otherwise, $i\in \intvl 1m - O-I$, and  $\biglambda{C}{O}{I}$ can be written $\biglambda{C}{O}{I}=\row{A \n L_i \n B}$. Hence 
\begin{equation*}
 \row{\test_i\n X \n \biglambda{C}{O}{I}}
= \row{\ch{\test_i}\n X \n A \n \ch{L_i} \n B}
\rp{lock}{c}\ \dead
\end{equation*}
\end{proof}\fi

\subsubsection{Variable}
In the following two sections, we assume that $p_\Lambda$ and $m$ are two fixed integers, and we define the $\literals$ gadget
$(\key_1,\ldots,\key_m,\test_1,\ldots,\test_m, \Lambda) = \literals{p_\Lambda}{m}$. Thus, we can use elements $\key_i$ and $\test_i$ for $i\in\intvl 1m$, and sequences $\biglambda{C}{O}{I}$ for any disjoint subsets $O$ and $I$ of $\intvl 1m$.

We now define a gadget simulating a boolean variable $x_i$. It holds two series 
of $\key$ elements: the ones with indices in $P$ (resp. $N$) open the locks corresponding 
to literals of the form $x_i$ (resp.~$\neg x_i$). 
When the triggering element, $\nu$, is brought to the head, a choice has to be made
between $P$ and $N$, and the locks associated with the chosen set (and only them) are open.

\begin{defn}\label{def:var}
Let $P, N$ be two disjoint subsets of $\intvl 1m$ ($P=\{\litp 1,\litp 2,\ldots,\litp q\}$, $N=\{\litn 1,\litn 2,\ldots,\litn {q'}\}$) 
 and $p$ be an integer, $\var{P}{N}{p}$ is defined by
\begin{align*}
 \var{P}{N}{p}&=(\nu, V, D)\\
\where
 (\take,\puth, G,H)&= \hook{p+2} \\
 (E,F)&= \fork{p+14}
\\
\preput{in}
 \nu&=\take \\
 V&=\row{G \n E \n \key_{\litp 1} \n \ldots \n \key_{\litp q} \n \puth \n \key_{\litn 1} \n \ldots \n \key_{\litn {q'}}\n F\n H}\\
 D&=\dock{p+2}{p+29}
 \end{align*}
Given a variable gadget $(\nu, V, D)=\var{P}{N}{p}$, we write
\begin{eqnarray*}
  V^1&=& \row{G'' \n \key_{\litn 1} \n \ldots \n \key_{\litn {q'}}\n F^1\n H''}\\
 V^2&=& \row{G'' \n \key_{\litp q} \n \ldots \n \key_{\litp {1}}\n F^2\n H''} 
\end{eqnarray*}
where  $G''$, $H''$, $F^1$, $F^2$, come from the definitions of Hook (Definition~\ref{def:hook}) and Fork (Definition~\ref{def:fork}).
\end{defn}

The following property determines the possible behavior of a variable gadget. \fullOnly{It is illustrated by Figure~\ref{fig:variable}.}

\ifisFullVersion\begin{figure}
\centering
 \includegraphics{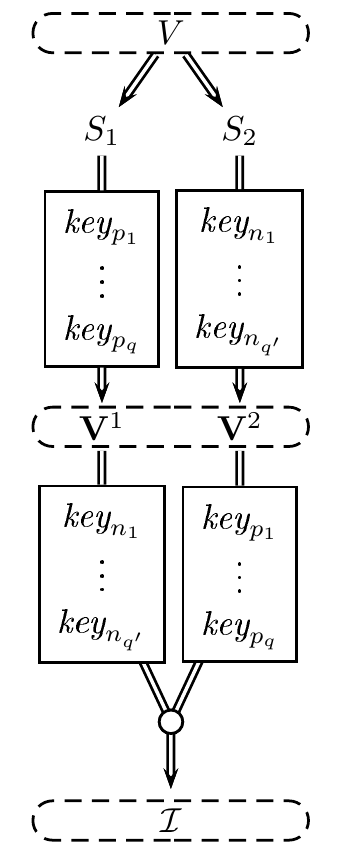}
\caption{\label{fig:variable} 
Initially, a variable gadget contains mainly the sequence $V$. Property \ref{prop:variable}{a} proves that two paths are possible, leading to sequences containing either $V^1$ or $V^2$. Along the first (resp. second) path, the locks with indices in $P$ (resp. $N$) are opened. By Property \ref{prop:variable}{b} (resp. c), there exists a path transforming $V^1$ (resp. $V^2$) into the identity over \intvl{p+1}{p+31}, which opens the remaining locks. 
}
\end{figure}\fi

\begin{property}\label{prop:variable}
Let 
$P$, $N$ be two disjoint subsets of $\intvl 1m$, $p$ be an integer, $X$ and $Y$ be two sequences, $O$, $I$ be two disjoint subsets of $\intvl 1m$, and $(\nu, V, D)=\var{P}{N}{p}$. For sub-property (a.) we require that $(P\cup N)\cap (O\cup I)=\emptyset$, for (b.) that $N\cap (O\cup I)=\emptyset$, and for (c.) that $P\cap (O\cup I)=\emptyset$ (these conditions are in fact necessarily satisfied by construction since all sequences considered are permutations)\wehave
\fullOnly{
\begin{equation*}
\subprop{a} 
\stack{\nu\n X\n V\n Y\n \biglambda{C}{O}{I}} \path
\left\{
  \stack{ X\n V^1 \n Y\n \biglambda{C-P}{O\cup P}{I}},
  \stack{ X\n V^2 \n Y\n \biglambda{C-N}{O\cup N}{I}}
\right\}
\qquad \subprop b
\stack{V^1\n X\n D \n Y \n \biglambda{C}{O}{I}} \path
\stack{X\n \ident{p+1}{p+31} \n Y \n \biglambda{C-N}{O\cup N}{I}} 
\qquad \subprop c
\stack{V^2\n X\n D \n Y \n \biglambda{C}{O}{I}} \path
\stack{X\n \ident{p+1}{p+31} \n Y \n \biglambda{C-P}{O\cup P}{I}} 
\end{equation*}
}
\abstractOnly{
\begin{align*}
\subprop{a} &
\row{\nu\n X\n V\n Y\n \biglambda{C}{O}{I}} \path \left.
\begin{cases}
  \row{ X\n V^1 \n Y\n \biglambda{C-P}{O\cup P}{I}},\\
  \row{ X\n V^2 \n Y\n \biglambda{C-N}{O\cup N}{I}}
\end{cases}\right\}
\\
\subprop b&
\row{V^1\n X\n D \n Y \n \biglambda{C}{O}{I}} \path
\row{X\n \ident{p+1}{p+31} \n Y \n \biglambda{C-N}{O\cup N}{I}} 
\\
\subprop c&
\row{V^2\n X\n D \n Y \n \biglambda{C}{O}{I}} \path
\row{X\n \ident{p+1}{p+31} \n Y \n \biglambda{C-P}{O\cup P}{I}} 
\end{align*}
}

\end{property}

\ifisFullVersion\begin{proof}
 \begin{align*}
  & \subprop{a} 
\row{\nu\n X\n V\n Y\n \biglambda{C}{O}{I}} &\hspace{8em}
\\
& = 
\row{\ch{\take} \n X\n \ch{G }\n E \n \key_{\litp 1} \n \ldots \n \key_{\litp q} \n \puth \n \key_{\litn 1} \n \ldots \n \key_{\litn {q'}}\n F\n \ch{H}\n Y\n \biglambda{C}{O}{I}} 
\\
& \rp{hook}{a}\path 
\row{\ch{E} \n \key_{\litp 1} \n \ldots \n \key_{\litp q} \n \puth \n \key_{\litn 1} \n \ldots \n \key_{\litn {q'}}\n \ch{F}\n G'\n \rev X\n H'\n Y\n \biglambda{C}{O}{I}} 
\\
& \rp{fork}{a}\path \{S_1,S_2\} \mbox{ \quad (where sequences $S_1$ and $S_2$ are described below)}
 \end{align*}
First,
\begin{align*}
 &S_1 = 
\row{\ch{\key_{\litp 1}} \n \key_{\litp 2} \n \ldots \n \key_{\litp q} \n \puth \n \key_{\litn 1} \n \ldots \n \key_{\litn {q'}}\n F^1\n G'\n \rev X\n H'\n Y\n \ch{\biglambda{C}{O}{I}}}  &\hspace{8em}\\
& \rp{literals}{a}\path 
\row{\ch{\key_{\litp 2}} \n \ldots \n \key_{\litp q} \n \puth \n \key_{\litn 1} \n \ldots \n \key_{\litn {q'}}\n F^1\n G'\n \rev X\n H'\n Y\n \ch{\biglambda{C-\{\litp 1\}}{O\cup\{\litp 1\}}{I}}}
\\
&\quad\vdots\\ 
& \rp{literals}{a}\path 
\row{\ch{\puth}\n \key_{\litn 1} \n \ldots \n \key_{\litn {q'}}\n F^1\n \ch{G'}\n \rev X\n \ch{H'}\n Y\n \biglambda{C-P}{O\cup P}{I}}
\\
& \rp{hook}{b}\path 
\row{X\n G''\n \key_{\litn 1} \n \ldots \n \key_{\litn {q'}}\n F^1\n H''\n Y\n \biglambda{C-P}{O\cup P}{I}}
\\
&=\row{ X\n V^1 \n Y\n \biglambda{C-P}{O\cup P}{I}}
\end{align*}
Second,
\begin{align*}
 &S_2 = 
\row{\ch{\key_{\litn {q'}}}\n \key_{\litn {q'-1}} \n \ldots \n\key_{\litn 1}   \n \puth \n \key_{\litp q}\n \ldots \n \key_{\litp 1}\n F^2\n G'\n \rev X\n H'\n Y\n \ch{\biglambda{C}{O}{I}}}  &\hspace{8em}
\\
& \rp{literals}{a}\path 
\row{\ch{\key_{\litn {q'-1}}} \n \ldots \n\key_{\litn 1}   \n \puth \n \key_{\litp q}\n \ldots \n \key_{\litp 1}\n F^2\n G'\n \rev X\n H'\n Y\n \ch{\biglambda{C-\{\litn {q'}\}}{O\cup\{\litn {q'}\}}{I}}}
\\
&\quad\vdots\\ 
& \rp{literals}{a}\path 
\row{\ch{\puth}\n \key_{\litp q}\n \ldots \n \key_{\litp 1}\n F^2\n \ch{G'}\n \rev X\n \ch{H'}\n Y\n \biglambda{C-N}{O\cup N}{I}}
\\
& \rp{hook}{b}\path 
\row{X\n G''\n \key_{\litp q}\n \ldots \n \key_{\litp 1}\n F^2\n H''\n Y\n \biglambda{C-N}{O\cup N}{I}}
\\
&=\row{ X\n V^2 \n Y\n \biglambda{C-N}{O\cup N}{I}}
\end{align*}

\begin{align*}
 & \subprop b \row{ V^1 \n X \n D \n Y\n \biglambda{C}{O}{I}}  &\hspace{15em}
\\
&= \row{\ch{G''}\n \key_{\litn 1} \n \ldots \n \key_{\litn {q'}}\n F^1\n \ch{H''}\n X\n D \n Y\n \biglambda{C}{O}{I}}
\\
&\rp{hook}{c}
\path \row{\ch{\key_{\litn 1} }\n\key_{\litn 2} \n \ldots \n \key_{\litn {q'}}\n F^1\n \rev \ident{p+3}{p+14}\n X\n D \n Y\n \ch{\biglambda{C}{O}{I}}}
\\
&\rp{literals}{a}
\path \row{\ch{\key_{\litn 2} } \n \ldots \n \key_{\litn {q'}}\n F^1\n \rev \ident{p+3}{p+14}\n X\n D \n Y\n \ch{\biglambda{C-\{\litn 1\}}{O\cup\{\litn 1\}}{I}}}
\\
&\quad\vdots\\ 
& \rp{literals}{a}\path \row{\ch{F^1}\n \rev \ident{p+3}{p+14}\n X\n D \n Y\n \biglambda{C-N}{O\cup N}{I}}
\\
& \rp{fork}{b}\path 
\row{\ch{\rev \ident{p+15}{p+29}}\n \ch{\rev \ident{p+3}{p+14}}\n X\n\ch{ D} \n Y\n \biglambda{C-N}{O\cup N}{I}}
\\
& \rp{dock}{}\path 
\row{X\n \ident{p+1}{p+31} \n Y\n \biglambda{C-N}{O\cup N}{I}}
\end{align*}

\begin{align*}
 & \subprop c \row{ V^2 \n X \n D \n Y\n \biglambda{C}{O}{I}}  &\hspace{15em}
\\
&= \row{\ch{G''}\n \key_{\litp q}\n \ldots \n \key_{\litp 1}\n F^2\n \ch{H''}\n X\n D \n Y\n \biglambda{C}{O}{I}}
\\
&\rp{hook}{c}
\path \row{\ch{\key_{\litp q}} \n \key_{\litp {q-1}}\n \ldots \n \key_{\litp 1}\n F^2\n \rev \ident{p+3}{p+14}\n X\n D \n Y\n \ch{\biglambda{C}{O}{I}}}
\\
&\rp{literals}{a}
\path \row{\ch{\key_{\litp {q-1}}}\n \ldots \n \key_{\litp 1}\n F^2\n \rev \ident{p+3}{p+14}\n X\n D \n Y\n \ch{\biglambda{C-\{\litp q\}}{O\cup\{\litp q\}}{I}}}
\\
&\quad\vdots\\ 
& \rp{literals}{a}\path \row{\ch{F^2}\n \rev \ident{p+3}{p+14}\n X\n D \n Y\n \biglambda{C-P}{O\cup P}{I}}
\\
& \rp{fork}{c}\path 
\row{\ch{\rev \ident{p+15}{p+29}}\n \ch{\rev \ident{p+3}{p+14}}\n X\n\ch{ D} \n Y\n \biglambda{C-P}{O\cup P}{I}}
\\
& \rp{dock}{}\path 
\row{X\n \ident{p+1}{p+31} \n Y\n \biglambda{C-P}{O\cup P}{I}}
\end{align*}
\end{proof}\fi

\subsubsection{Clause}

The following gadget simulates a 3-clause in a boolean formula. It holds the $\test$ elements for three locks, corresponding to three literals. When the triggering element, $\gamma$, is at the head of a sequence, three distinct efficient paths may be followed. In each such path, one of the three locks is tested: in other words, any efficient path leading to the identity requires one of the locks to be open.

\begin{defn}\label{def:clause}
Let $a,b,c\in\intvl 1m$ be pairwise distinct integers and $p$ be an integer, 
$\clause{a}{b}{c}{p}$
is defined by
\begin{eqnarray*}
\clause{a}{b}{c}{p}&=&(\gamma,\Gamma,\Delta) \\
\where 
(E_1,F_1)&=&\fork{p+2}\\
(E_2,F_2)&=&\fork{p+45}\\
(\take_1,\puth_1,G_1,H_1)&=&\hook{p+21}\\
(\take_2,\puth_2,G_2,H_2)&=&\hook{p+33}\\
D_1&=&\dock{p+2}{p+17}\\
D_2&=&\dock{p+21}{p+60}\\
\mbox{in \quad}
 \gamma&=& \take_1  \\
 \Gamma&=&\row{ 
 G_1\n E_1 \n \take_2 \n \puth_1\n 
\test_c \n F_1 \n G_2 \n E_2 \n \test_a \n \puth_2 \n \test_b \n F_2 \n H_2\n H_1}\\
\Delta&=& \row{ D_1 \n D_2 } 
\end{eqnarray*}

Given a clause gadget $(\gamma,\Gamma,\Delta)=\clause{a}{b}{c}{p}$, we write
\begin{eqnarray*}
\Gamma^1&=&\row{G''_1 \n \test_c \n F_1^1 \n G''_2 \n \test_b \n F^1_2 \n H''_2\n H''_1}\\
\Gamma^2&=&\row{G''_1 \n \test_c \n F_1^1 \n G''_2 \n \test_a \n F^2_2 \n H''_2\n H''_1}\\
\Gamma^3&=&\row{G''_1 \n \take_2 \n F_1^2 \n G_2 \n E_2 \n \test_a \n \puth_2 \n \test_b \n F_2 \n H_2\n H''_1}
\end{eqnarray*}
\end{defn}

The following two properties determine the possible behavior of a clause gadget. \fullOnly{They are illustrated by Figure~\ref{fig:clauses}.}

\ifisFullVersion\begin{figure}
\centering
 \includegraphics{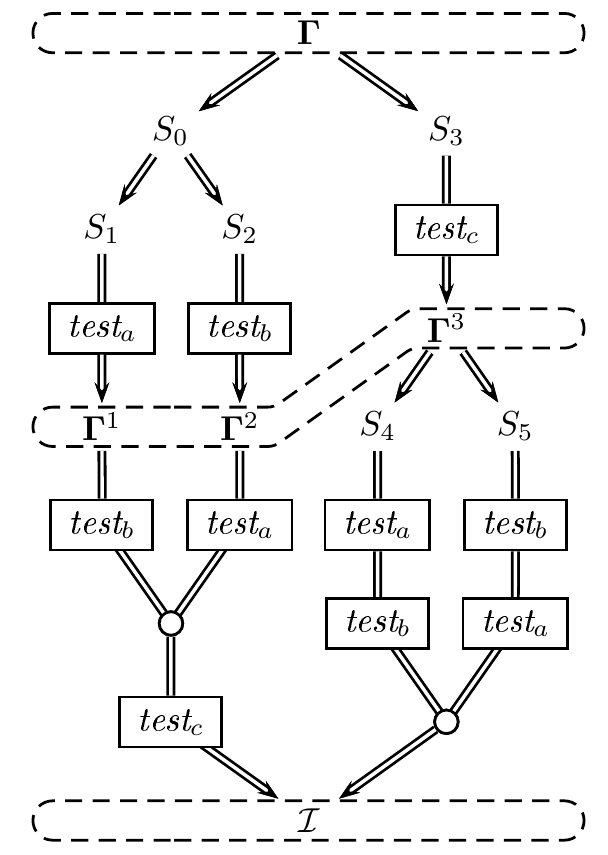}
\caption{\label{fig:clauses} 
Initially, a clause gadget contains mainly the sequence $\Gamma$. Property~\ref{prop:clause} proves that three paths may be possible, leading to sequences containing either $\Gamma^1$, $\Gamma^2$ or $\Gamma^3$. Because of the $\test$ elements, each path requires one lock to be open (either $a$, $b$ or $c$). By Property~\ref{prop:clause2}{a} (resp. b, c), there exists a path transforming $\Gamma^1$ (resp. $\Gamma^2$, $\Gamma^3$) into the identity over \intvl{p+1}{p+62}, provided the remaining locks are open. 
}
\end{figure}\fi

\begin{property}\label{prop:clause}
Let $X$ and $Y$ be any sequences, and $O,I$ be two disjoint subsets of $\intvl1m$\wehave
\begin{equation*}
\stackrow{ \gamma \n X \n \Gamma \n Y  \n\biglambda{C}{O}{I} } \path\mathbb T
 \end{equation*}

where $\mathbb T$ contains from 0 to 3 sequences, and is defined by:
\ifisFullVersion
\begin{equation*}
\stack{  X \n \Gamma^1\n Y   \n \biglambda{C}{O-\{a\}}{I\cup\{a\}} } \in \mathbb T \mbox{ iff } a\in O
\qquad\quad
\stackrow{  X \n \Gamma^2\n Y   \n \biglambda{C}{O-\{b\}}{I\cup\{b\}} } \in \mathbb T \mbox{ iff } b\in O
\qquad\quad
\stackrow{  X \n \Gamma^3\n Y   \n \biglambda{C}{O-\{c\}}{I\cup\{c\}} } \in \mathbb T \mbox{ iff } c\in O
\end{equation*}
\fi

\end{property}

\ifisFullVersion\begin{proof}
\begin{align*}
&\row{ \gamma \n X \n \Gamma \n Y  \n\biglambda{C}{O}{I} }&&
\\
&=  \row{\ch{ \take_1} \n X\n 
 \ch{G_1}\n E_1 \n \take_2 \n \puth_1\n 
\test_c \n F_1 \n G_2 \n E_2 \n \test_a \n \puth_2 \n \test_b \n F_2 \n H_2\n\ch{ H_1}
\n Y  \n\biglambda{C}{O}{I} }
\\
&\rp{hook}{a}\path
  \row{  
 \ch{E_1} \n \take_2 \n \puth_1\n 
\test_c \n \ch{F_1} \n G_2 \n E_2 \n \test_a \n \puth_2 \n \test_b \n F_2 \n H_2\n 
G'_1\n \rev X \n H'_1
\n Y  \n\biglambda{C}{O}{I} }
\\
&\rp{fork}{a}\path \{S_0,S_3\}
\end{align*}
\begin{align*}
 S_0&=
  \row{  \ch{ \take_2 }\n \puth_1\n \test_c \n F_1^1 
\n \ch{G_2} \n E_2 \n \test_a \n \puth_2 \n \test_b \n F_2 \n \ch{H_2}\n 
G'_1\n \rev X \n H'_1
\n Y  \n\biglambda{C}{O}{I} } &&\quad
\\
&\rp{hook}{a}\path 
 \row{   \ch{E_2 }\n \test_a \n \puth_2 \n \test_b \n\ch{ F_2 }           
\n G_2' \n \rev F_1^1 \n \test_c \n \puth_1\n H_2'\n 
G'_1\n \rev X \n H'_1
\n Y  \n\biglambda{C}{O}{I} } 
\\
&\rp{fork}{a}\path \{ S_1, S_2\}  &&\quad
\end{align*}
\begin{align*}
S_1&=
 \row{ \ch{  \test_a} \n \puth_2 \n \test_b \n F^1_2            
\n G_2' \n \rev F_1^1 \n \test_c \n \puth_1\n H_2'\n 
G'_1\n \rev X \n H'_1
\n Y  \n\ch{\biglambda{C}{O}{I} }}
\\
\mbox{if } &a \notin  O \mbox{ then } S_1\rp{literals}{c}\ \dead \\
\mbox{if } &a \in  O \mbox{ then } \\
S_1 &\rp{literals}{b}\path 
 \row{  \ch{ \puth_2} \n \test_b \n F^1_2            
\n\ch{ G_2'} \n \rev F_1^1 \n \test_c \n \puth_1\n \ch{H_2'}\n 
G'_1\n \rev X \n H'_1
\n Y  \n\biglambda{C}{O-\{a\}}{I\cup\{a\}} } 
\\
&\rp{hook}{b}\path 
 \row{ \ch{\puth_1} \n \test_c \n  F_1^1            
\n G''_2 \n\test_b \n F^1_2 \n H''_2\n 
\ch{G'_1}\n \rev X \n \ch{H'_1}
\n Y  \n\biglambda{C}{O-\{a\}}{I\cup\{a\}} }
\\
&\rp{hook}{b}\path 
 \row{ X \n G''_1 \n
 \test_c \n  F_1^1            
\n G''_2 \n\test_b \n F^1_2 \n H''_2\n 
H''_1
\n Y  \n\biglambda{C}{O-\{a\}}{I\cup\{a\}} }
\\
& =  \row{X \n \Gamma^1\n Y   \n \biglambda{C}{O-\{a\}}{I\cup\{a\}} }&&\quad
\end{align*}
\begin{align*}
S_2&=
\row{  \ch{ \test_b} \n \puth_2 \n \test_a \n F^2_2            
\n G_2' \n \rev F_1^1 \n \test_c \n \puth_1\n H_2'\n 
G'_1\n \rev X \n H'_1
\n Y  \n\ch{\biglambda{C}{O}{I}} } 
\\
\mbox{if } &b \notin  O \mbox{ then } S_2\rp{literals}{c}\ \dead \\
\mbox{if } &b \in  O \mbox{ then } \\
S_2&\rp{literals}{b}\path
\row{   \ch{\puth_2} \n \test_a \n F^2_2            
\n\ch{ G_2'} \n \rev F_1^1 \n \test_c \n \puth_1\n \ch{H_2'}\n 
G'_1\n \rev X \n H'_1
\n Y  \n\biglambda{C}{O-\{b\}}{I\cup\{b\}} } 
\\
&\rp{hook}{b}\path
\row{\ch{\puth_1}  \n \test_c \n  F_1^1 \n
G''_2  \n \test_a \n F^2_2            
\n H''_2\n 
\ch{G'_1}\n \rev X \n \ch{H'_1}
\n Y  \n\biglambda{C}{O-\{b\}}{I\cup\{b\}} } 
\\
&\rp{hook}{b}\path
\row{ X\n   G''_1 \n
\test_c \n  F_1^1 \n
G''_2 \n \test_a \n F^2_2            
\n H''_2\n 
H''_1
\n Y  \n\biglambda{C}{O-\{b\}}{I\cup\{b\}} }
\\
& =  \row{X \n \Gamma^2\n Y   \n \biglambda{C}{O-\{b\}}{I\cup\{b\}} }&&\quad
\end{align*}
\begin{align*}
  S_3&= \row{    \ch{ \test_c} \n \puth_1\n\take_2\n F_1^2 
\n G_2 \n E_2 \n \test_a \n \puth_2 \n \test_b \n F_2 \n H_2\n 
G'_1\n \rev X \n H'_1
\n Y  \n\ch{\biglambda{C}{O}{I}} } 
\\
\mbox{if } &c \notin  O \mbox{ then } S_3\rp{literals}{c}\ \dead \\
\mbox{if } &c \in  O \mbox{ then } \\
S_3& \rp{literals}{b}\path 
   \row{   \ch{ \puth_1}\n\take_2\n F_1^2 
\n G_2 \n E_2 \n \test_a \n \puth_2 \n \test_b \n F_2 \n H_2\n 
\ch{G'_1}\n \rev X \n \ch{H'_1}
\n Y  \n\biglambda{C}{O-\{c\}}{I\cup\{c\}} } 
\\
&\rp{hook}{b}\path 
   \row{  X \n G''_1\n  \take_2\n F_1^2 
\n G_2 \n E_2 \n \test_a \n \puth_2 \n \test_b \n F_2 \n H_2\n 
H''_1
\n Y  \n\biglambda{C}{O-\{c\}}{I\cup\{c\}} } 
\\ 
&=  \row{X \n \Gamma^3\n Y   \n \biglambda{C}{O-\{c\}}{I\cup\{c\}} }&&\quad
\end{align*}

\end{proof}\fi

\begin{property}\label{prop:clause2}
Let $Y$ and $Z$ be any sequences, and $O,I$ be two disjoint subsets of $\intvl1m$\wehave

\begin{eqnarray*}
\subprop a \mbox{If } b,c\in O ,\mbox{ then }\ \ \fstackrow{ \Gamma^1\n Y\n
 \Delta\n Z\n \biglambda{C}{O}{I} } 
&\path &
\fstackrow{Y \n \ident{p+1}{p+62}\n Z\n \biglambda{C}{O-\{b,c\}}{I\cup\{b,c\}} } \\
\subprop b \mbox{If } a,c\in O ,\mbox{ then }\ \ \fstackrow{ \Gamma^2\n Y\n
 \Delta\n Z\n \biglambda{C}{O}{I} } 
&\path &
\fstackrow{Y \n \ident{p+1}{p+62}\n Z\n \biglambda{C}{O-\{a,c\}}{I\cup\{a,c\}} } \\
\subprop c \mbox{If } a,b\in O ,\mbox{ then }\ \ \fstackrow{ \Gamma^3\n Y\n
 \Delta\n Z\n \biglambda{C}{O}{I} } 
&\path &
\fstackrow{Y \n \ident{p+1}{p+62}\n Z\n \biglambda{C}{O-\{a,b\}}{I\cup\{a,b\}} } 
\end{eqnarray*}
\end{property}

\ifisFullVersion\begin{proof}
\begin{align*} 
&\subprop a\row{ \Gamma^1\n Y\n \Delta\n Z\n \biglambda{C}{O}{I} }
\\
&= \row{\ch{G''_1} \n \test_c \n F_1^1 \n G''_2 \n \test_b \n F^1_2 \n H''_2\n \ch{H''_1}\n Y\n D_1\n D_2\n  Z\n \biglambda{C}{O}{I}}
\\
&\rp{hook}{c}\path
\row{ \ch{\test_c} \n F_1^1 \n G''_2 \n \test_b \n F^1_2 \n H''_2\n \rev\ident{p+22}{p+33} \n Y\n D_1\n D_2\n  Z\n \ch{\biglambda{C}{O}{I}}}
\\
&\rp{literals}{b}\path
\row{ \ch{F_1^1} \n G''_2 \n \test_b \n F^1_2 \n H''_2\n \rev\ident{p+22}{p+33} \n Y\n D_1\n D_2\n  Z\n \biglambda{C}{O-\{c\}}{I\cup\{c\}}}
\\
&\rp{fork}{b}\path
\row{ \ch{\rev\ident{p+3}{p+17}} \n G''_2 \n \test_b \n F^1_2 \n H''_2\n \rev\ident{p+22}{p+33} \n Y\n \ch{D_1}\n D_2\n  Z\n \biglambda{C}{O-\{c\}}{I\cup\{c\}}}
\\
&\rp{dock}{}\path
\row{\ch{ G''_2} \n \test_b \n F^1_2 \n \ch{H''_2}\n \rev\ident{p+22}{p+33}  \n Y\n \ident{p+1}{p+19}\n D_2\n  Z\n \biglambda{C}{O-\{c\}}{I\cup\{c\}}}
\\
&\rp{hook}{c}\path
\row{ \ch{\test_b} \n F^1_2 \n \rev\ident{p+34}{p+45} \n  \rev\ident{p+22}{p+33} \n Y\n \ident{p+1}{p+19}\n D_2\n  Z\n \ch{\biglambda{C}{O-\{c\}}{I\cup\{c\}}}}
\\
&\rp{literals}{b}\path
\row{ \ch{F^1_2} \n  \rev\ident{p+34}{p+45} \n \rev\ident{p+22}{p+33} \n Y\n \ident{p+1}{p+19}\n D_2\n  Z\n \biglambda{C}{O-\{b,c\}}{I\cup\{b,c\}}}
\\
&\rp{fork}{b}\path
\row{ \ch{\rev\ident{p+46}{p+60}} \n \ch{\rev\ident{p+34}{p+45}} \n \ch{\rev\ident{p+22}{p+33}}
 \n Y\n \ident{p+1}{p+19}\n \ch{D_2}\n  Z\n \biglambda{C}{O-\{b,c\}}{I\cup\{b,c\}}}
\\
&\rp{dock}{}\path
\row{  Y\n \ident{p+1}{p+19}\n \ident{p+20}{p+62}
\n Z\n \biglambda{C}{O-\{b,c\}}{I\cup\{b,c\}}}
\\
&=
\row{  Y\n \ident{p+1}{p+62}\n Z\n \biglambda{C}{O-\{b,c\}}{I\cup\{b,c\}}}
\end{align*}




\begin{align*} 
&\subprop b\row{ \Gamma^2\n Y\n \Delta\n Z\n \biglambda{C}{O}{I} }
\\
&= \row{\ch{G''_1} \n \test_c \n F_1^1 \n G''_2 \n \test_a \n F^2_2 \n H''_2\n \ch{H''_1}\n Y\n D_1\n D_2\n  Z\n \biglambda{C}{O}{I}}
\\
&\rp{hook}{c}\path
\row{ \ch{\test_c} \n F_1^1 \n G''_2 \n \test_a \n F^2_2 \n H''_2\n \rev\ident{p+22}{p+33} \n Y\n D_1\n D_2\n  Z\n \ch{\biglambda{C}{O}{I}}}
\\
&\rp{literals}{b}\path
\row{ \ch{F_1^1} \n G''_2 \n \test_a \n F^2_2 \n H''_2\n \rev\ident{p+22}{p+33} \n Y\n D_1\n D_2\n  Z\n \biglambda{C}{O-\{c\}}{I\cup\{c\}}}
\\
&\rp{fork}{b}\path
\row{ \ch{\rev\ident{p+3}{p+17}} \n G''_2 \n \test_a \n F^2_2 \n H''_2\n \rev\ident{p+22}{p+33} \n Y\n \ch{D_1}\n D_2\n  Z\n \biglambda{C}{O-\{c\}}{I\cup\{c\}}}
\\
&\rp{dock}{}\path
\row{\ch{ G''_2} \n \test_a \n F^2_2 \n \ch{H''_2}\n \rev\ident{p+22}{p+33}  \n Y\n \ident{p+1}{p+19}\n D_2\n  Z\n \biglambda{C}{O-\{c\}}{I\cup\{c\}}}
\\
&\rp{hook}{c}\path
\row{ \ch{\test_a} \n F^2_2 \n \rev\ident{p+34}{p+45} \n  \rev\ident{p+22}{p+33} \n Y\n \ident{p+1}{p+19}\n D_2\n  Z\n \ch{\biglambda{C}{O-\{c\}}{I\cup\{c\}}}}
\\
&\rp{literals}{b}\path
\row{ \ch{F^2_2} \n  \rev\ident{p+34}{p+45} \n \rev\ident{p+22}{p+33} \n Y\n \ident{p+1}{p+19}\n D_2\n  Z\n \biglambda{C}{O-\{a,c\}}{I\cup\{a,c\}}}
\\
&\rp{fork}{c}\path
\row{ \ch{\rev\ident{p+46}{p+60}} \n  \ch{\rev\ident{p+34}{p+45}} \n \ch{\rev\ident{p+22}{p+33}} \n Y\n \ident{p+1}{p+19}\n \ch{D_2}\n  Z\n \biglambda{C}{O-\{a,c\}}{I\cup\{a,c\}}}
\\
&\rp{dock}{}\path
\row{  Y\n \ident{p+1}{p+19}\n \ident{p+20}{p+62}
\n Z\n \biglambda{C}{O-\{a,c\}}{I\cup\{a,c\}}}
\\
&=
\row{  Y\n \ident{p+1}{p+62}\n Z\n \biglambda{C}{O-\{a,c\}}{I\cup\{a,c\}}}
\end{align*}

\begin{align*}
 &\subprop c\row{\Gamma^3\n Y  \n \Delta \n Z\n \biglambda{C}{O}{I} }
\\
&=
 \row{  \ch{G''_1}\n  \take_2\n F_1^2 
\n G_2 \n E_2 \n \test_a \n \puth_2 \n \test_b \n F_2 \n H_2\n 
\ch{H''_1}
\n Y\n D_1\n D_2\n  Z \n\biglambda{C}{O}{I} } 
\\ 
&\rp{hook}{c}\path
 \row{  \ch{\take_2}\n F_1^2 
\n \ch{G_2} \n E_2 \n \test_a \n \puth_2 \n \test_b \n F_2 \n \ch{H_2}\n 
\rev \ident{p+22}{p+33}
\n Y\n D_1\n D_2\n  Z \n\biglambda{C}{O}{I} } 
\\ 
&\rp{hook}{a}\path
 \row{ \ch{E_2 }\n \test_a \n \puth_2 \n \test_b \n \ch{F_2} \n 
G_2' \n\rev F_1^2 \n  H_2'\n 
\rev \ident{p+22}{p+33}
\n Y\n D_1\n D_2\n  Z \n\biglambda{C}{O}{I} } 
\\ 
&\rp{fork}{a}\path \{S_4, S_5 \}
\\ 
&S_4 = 
 \row{ \ch{\test_a} \n \puth_2 \n \test_b \n F_2^1 \n 
G_2' \n\rev F_1^2 \n  H_2'\n 
\rev \ident{p+22}{p+33}
\n Y\n D_1\n D_2\n  Z \n\ch{\biglambda{C}{O}{I}} } 
\\
&\rp{literals}{b}\path 
 \row{\ch{ \puth_2 }\n \test_b \n F_2^1 \n 
\ch{G_2'} \n\rev F_1^2 \n  \ch{H_2'}\n 
\rev \ident{p+22}{p+33}
\n Y\n D_1\n D_2\n  Z \n\biglambda{C}{O-\{a\}}{I\cup\{a\}} } 
\\
&\rp{hook}{b}\path 
 \row{ \ch{F_1^2} \n G_2''\n
\test_b \n F_2^1 \n 
 H_2''\n 
\rev \ident{p+22}{p+33}
\n Y\n D_1\n D_2\n  Z \n\biglambda{C}{O-\{a\}}{I\cup\{a\}} } 
\\
&\rp{fork}{c}\path 
 \row{\ch{ \rev \ident{p+3}{p+17} }\n G_2''\n
\test_b \n F_2^1 \n 
 H_2''\n 
\rev \ident{p+22}{p+33}
\n Y\n \ch{D_1}\n D_2\n  Z \n\biglambda{C}{O-\{a\}}{I\cup\{a\}} } 
\\
&\rp{dock}{}\path 
 \row{ \ch{ G_2''}\n
\test_b \n F_2^1 \n 
\ch{ H_2''}\n 
\rev \ident{p+22}{p+33}
\n Y\n \ident{p+1}{p+19} \n D_2\n  Z \n\biglambda{C}{O-\{a\}}{I\cup\{a\}} } 
\\
&\rp{hook}{c}\path 
 \row{  \ch{\test_b} \n F_2^1 \n 
\rev  \ident{p+34}{p+45}\n 
\rev \ident{p+22}{p+33}
\n Y\n \ident{p+1}{p+19} \n D_2\n  Z \n\ch{\biglambda{C}{O-\{a\}}{I\cup\{a\}}} } 
\\
&\rp{literals}{b}\path 
 \row{ \ch{F_2^1} \n 
\rev  \ident{p+34}{p+45}\n 
\rev \ident{p+22}{p+33}
\n Y\n \ident{p+1}{p+19} \n D_2\n  Z \n\biglambda{C}{O-\{a,b\}}{I\cup\{a,b\}} } 
\\
&\rp{fork}{b}\path 
 \row{\ch{\rev \ident{p+46}{p+60}} \n 
\ch{\rev  \ident{p+34}{p+45}}\n 
\ch{\rev \ident{p+22}{p+33}}
\n Y\n \ident{p+1}{p+19} \n \ch{D_2}\n  Z \n\biglambda{C}{O-\{a,b\}}{I\cup\{a,b\}} } 
\\
&\rp{dock}{}\path 
 \row{   Y\n \ident{p+1}{p+19} \n \ident{p+20}{p+62} 
\n Z \n\biglambda{C}{O-\{a,b\}}{I\cup\{a,b\}} } 
\\
& = 
 \row{   Y\n \ident{p+1}{p+62}\n Z \n\biglambda{C}{O-\{a,b\}}{I\cup\{a,b\}} } 
\\
\end{align*}
\begin{align*}
&S_5 = 
 \row{\ch{ \test_b} \n \puth_2 \n \test_a \n F_2^2 \n 
G_2' \n\rev F_1^2 \n  H_2'\n 
\rev \ident{p+22}{p+33}
\n Y\n D_1\n D_2\n  Z \n\ch{\biglambda{C}{O}{I}} } 
\\
&\rp{literals}{b}\path 
 \row{ \ch{\puth_2 }\n \test_a \n F_2^2 \n 
\ch{G_2'} \n\rev F_1^2 \n  \ch{H_2'}\n 
\rev \ident{p+22}{p+33}
\n Y\n D_1\n D_2\n  Z \n\biglambda{C}{O-\{a\}}{I\cup\{a\}} } 
\\
&\rp{hook}{b}\path 
 \row{ \ch{F_1^2} \n G_2''\n
\test_a \n F_2^2 \n 
 H_2''\n 
\rev \ident{p+22}{p+33}
\n Y\n D_1\n D_2\n  Z \n\biglambda{C}{O-\{a\}}{I\cup\{a\}} } 
\\
&\rp{fork}{c}\path 
 \row{ \ch{\rev \ident{p+3}{p+17}} \n G_2''\n
\test_a \n F_2^2 \n 
 H_2''\n 
\rev \ident{p+22}{p+33}
\n Y\n\ch{ D_1}\n D_2\n  Z \n\biglambda{C}{O-\{a\}}{I\cup\{a\}} } 
\\
&\rp{dock}{}\path 
 \row{ \ch{ G_2''}\n
\test_a \n F_2^2 \n 
\ch{ H_2''}\n 
\rev \ident{p+22}{p+33}
\n Y\n \ident{p+1}{p+19} \n D_2\n  Z \n\biglambda{C}{O-\{a\}}{I\cup\{a\}} } 
\\
&\rp{hook}{c}\path 
 \row{  \ch{\test_a} \n F_2^2 \n 
\rev  \ident{p+34}{p+45}\n 
\rev \ident{p+22}{p+33}
\n Y\n \ident{p+1}{p+19} \n D_2\n  Z \n\ch{\biglambda{C}{O-\{a\}}{I\cup\{a\}}} } 
\\
&\rp{literals}{b}\path 
 \row{ \ch{F_2^2} \n 
\rev  \ident{p+34}{p+45}\n 
\rev \ident{p+22}{p+33}
\n Y\n \ident{p+1}{p+19} \n D_2\n  Z \n\biglambda{C}{O-\{a,b\}}{I\cup\{a,b\}} } 
\\
&\rp{fork}{c}\path 
 \row{\ch{\rev \ident{p+46}{p+60}} \n 
\ch{\rev  \ident{p+34}{p+45}}\n 
\ch{\rev \ident{p+22}{p+33}}
\n Y\n \ident{p+1}{p+19} \n \ch{D_2}\n  Z \n\biglambda{C}{O-\{a,b\}}{I\cup\{a,b\}} } 
\\
&\rp{dock}{}\path 
 \row{   Y\n \ident{p+1}{p+19} \n \ident{p+20}{p+62} 
\n Z \n\biglambda{C}{O-\{a,b\}}{I\cup\{a,b\}} } 
\\
& = 
 \row{   Y\n \ident{p+1}{p+62}\n Z \n\biglambda{C}{O-\{a,b\}}{I\cup\{a,b\}} } 
\\
\end{align*}

\end{proof}\fi

\subsection{Reduction}

Let $\phi$ be a boolean formula over $l$ variables in conjunctive normal form, such that each clause contains exactly three literals. We write $k$ the number of clauses, $m=3k$ the total number of literals, and $\{\lambda_1,\ldots,\lambda_m\}$ the set of literals. Let $n=31l + 62k + 12m$\fullOnly{ (thus, $n=31l+98k$)}. 

\begin{defn}\label{defn:sphi}
We define the sequence $S_\phi$ as the permutation of $\intvl 1n$ obtained by:
\begin{align*}
& (\key_1,\ldots,\key_m,\test_1,\ldots,\test_m, \Lambda) = \literals{31l+62k}{m} \\
& \mbox{For all } i \in \intvl 1l\\
& \quad P_i= \{j\in \intvl 1m \mid \lambda_j = x_i\} \\
& \quad N_i= \{j\in \intvl 1m \mid \lambda_j = \neg x_i\} \\
& \quad (\nu_i, V_i, D_i) = \var{P_i}{N_i}{31(i-1)} \\
& \mbox{For all } i \in \intvl 1k\\
& \quad (a_i,b_i,c_i)= \mbox{ indices such that the $i$-th clause of $\phi$ is } \lambda_{a_i}\vee \lambda_{b_i}\vee  \lambda_{c_i}\\
& \quad (\gamma_i, \Gamma_i, \Delta_i) = \clause{a_i}{b_i}{c_i}{31l+62(i-1)}\\
& S_\phi=
\row{
\nu_1,\ldots,\nu_l,
\gamma_1,\ldots,\gamma_k,
V_1,\ldots,V_l,
\Gamma_1,\ldots,\Gamma_k,
D_1,\ldots,D_l,
\Delta_1,\ldots,\Delta_k,
\biglambda{x}{\emptyset}{\emptyset}
}
\end{align*}
\end{defn}

Two things should be noted in this definition. First, elements $\key_i$ and $\test_i$ are used in the clause and variable gadgets, although they are not explicitly stated in the parameters (cf. Definitions~\ref{def:var} and~\ref{def:clause}). Second, one could assume that literals are sorted in the formula ($\phi=(\lambda_1\vee\lambda_2\vee\lambda_3)\wedge\dots$), so that $a_i=3i-2$, $b_i=3i-1$ and $c_i=3i$, but it is not necessary since these values are not used in the following.

We now aim at proving Theorem~\ref{thm:big} (p.~\pageref{thm:big}), which states that $S_\phi$ is efficiently sortable if and only if the formula $\phi$ is satisfiable. Several preliminary lemmas are necessary, and the overall process is illustrated in Figure~\ref{fig:general}.

\begin{figure}
 \centering
$ S_\phi=
\row{
\nu_1,\ldots,\nu_l,
\gamma_1,\ldots,\gamma_k,
V_1,\ldots,V_l,
\Gamma_1,\ldots,\Gamma_k,
D_1,\ldots,D_l,
\Delta_1,\ldots,\Delta_k,
\biglambda{x}{\emptyset}{\emptyset}
}$

\includegraphics{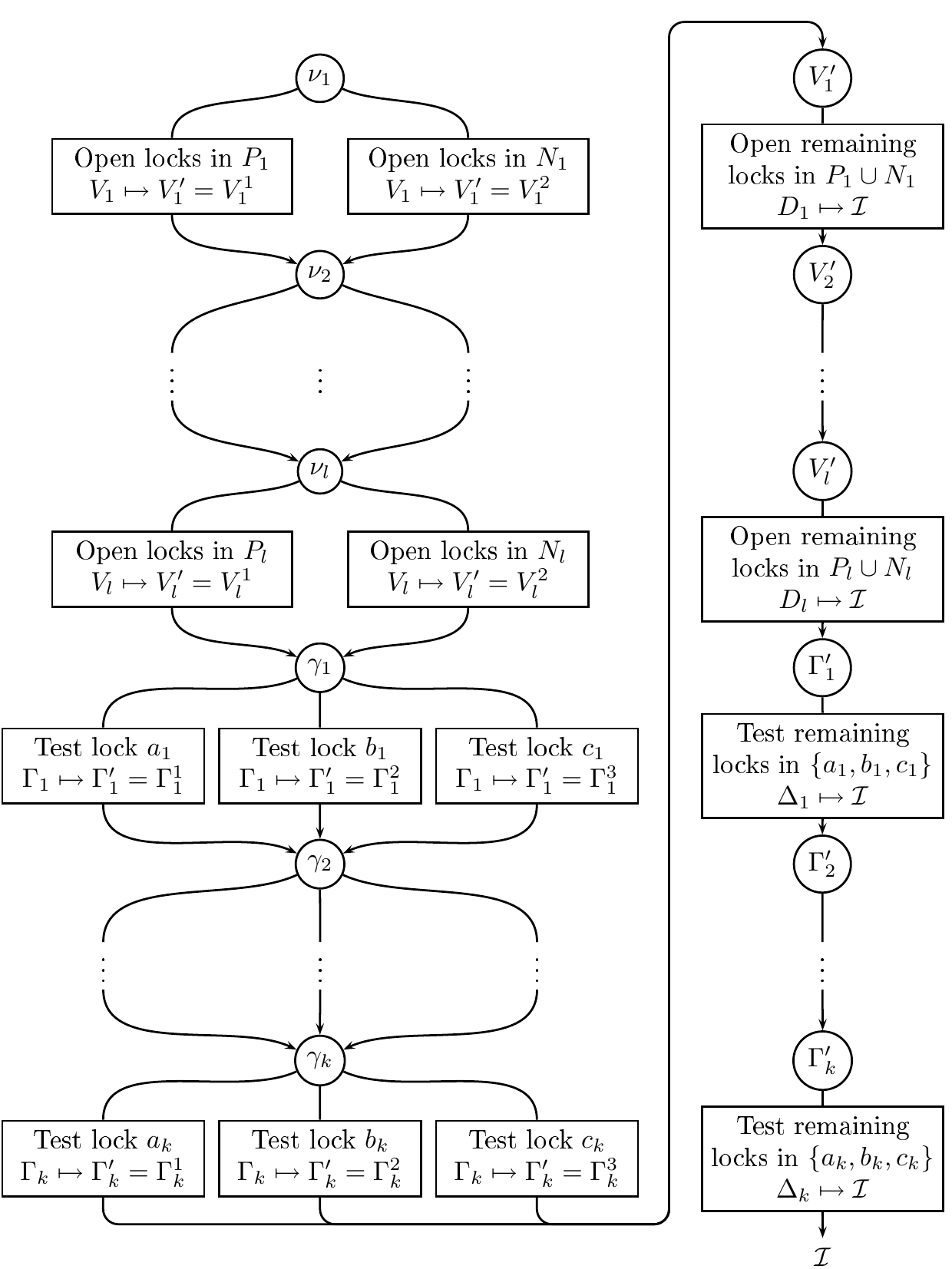}
\caption{\label{fig:general}Description of an efficient sorting of $S_\phi$. Circular nodes correspond to head elements or sequences especially relevant (landmarks). We start with the head element of $S_\phi$: $\nu_1$.
From each landmark, one, two or three paths are possible before reaching the next landmark, each path having its own effects, stated in rectangles, on the sequence. Possible effects are: transforming a subsequence of $S_\phi$ (symbol $\mapsto$), opening a lock, testing a lock (such a path requires the lock to be open).
The top-left quarter, from $\nu_1$ to $\nu_l$, is studied in Section~\ref{red:assignment};
the bottom-left quarter, from $\gamma_1$ to $\gamma_k$, is studied in Section~\ref{red:selection}; and the right half, from $V'_1$ to $\Gamma'_k$, is studied in Section~\ref{red:ending}. Indices are removed from identity sequences ($\mathcal I$) for readability.
}
\end{figure}

\subsubsection{Variable assignment}\label{red:assignment}

\ifisFullVersion     
\begin{defn}
Let $r\in \intvl 0l$.
An \emph{$r$-assignment} is a partition $\bigP=(T,F)$ of $\intvl 1r$. An $l$-assignment is called a \emph{full} assignment.
Using notations from Definition~\ref{defn:sphi}, we define the sequence $\SP{\bigP}$ by:

\begin{align*}
&\text{For all }i\in \intvl1r, \quad V_i' = 
\begin{cases} 
V_i^1\text{ if }i\in T\\
V_i^2\text{ if }i\in F\\
\end{cases}\\
&O=\bigcup_{i\in T} P_i \cup \bigcup_{i\in F} N_i\\
&\SP{\bigP}=
\row{
\nu_{r+1},\ldots,\nu_l,
\gamma_1,\ldots,\gamma_k,
V_1',\ldots,V_r',
V_{r+1},\ldots,V_l,\\
&\phantom{\SP{\bigP}=<}
\Gamma_1,\ldots,\Gamma_k,
D_1,\ldots,D_l,
\Delta_1,\ldots,\Delta_k,
\biglambda{xx}{O}{\emptyset}
}
\end{align*}
\end{defn}

\begin{property}\label{prop:1stepSP}
Let $r\in\intvl 0l$ with $r<l$, $\bigP=(T,F)$ be any $r$-assignment, $\bigP_1=(T\cup\{r+1\}, F)$ and  $\bigP_2=(T, F\cup\{r+1\})$.
Then  
\begin{equation*}
\SP{\bigP} \path
\left\{
   \SP{\bigP_1},
   \SP{\bigP_2}
\right\}
\end{equation*}
\end{property}
\begin{proof}
This is a direct application of Property~\ref{prop:variable}.a on variable $(\nu_{r+1},V_{r+1},D_{r+1})$, using sequences:
\begin{align*}
 X&=\row{
\nu_{r+2},\ldots,\nu_l,
\gamma_1,\ldots,\gamma_k,
V_1',\ldots,V_r'} \\
Y&=\row{
V_{r+2},\ldots,V_l,
\Gamma_1,\ldots,\Gamma_k,
D_1,\ldots,D_l,
\Delta_1,\ldots,\Delta_k}
\end{align*}
\end{proof}

\fi             

With the following lemma, we ensure that any sequence of efficient flips from $S_\phi$ begins with a full assignment of the boolean variables, and every possible assignment can be reached using only efficient flips.

\begin{lemma} \label{lem:fruVar}
\begin{equation*}
 S_\phi \path \left\{ \SP{\bigP} \mid \bigP \text{ full assignment} \right\}
\end{equation*}
\end{lemma}
\ifisFullVersion\begin{proof}
We prove $S_\phi \path \left\{ \SP{\bigP} \mid \bigP\ r-\text{assignment} \right\}$ by induction for all $r\in\intvl 0l$, and the lemma is deduced from the case $r=l$. 

There is only one 0-assignment, which is $\bigP_0=(\emptyset,\emptyset)$, and $S_\phi=\SP{\bigP_0}$.
Consider now any $r<l$. We use notations $\bigP_1$ and  $\bigP_2$ from Property~\ref{prop:1stepSP}. 
Then any $(r+1)$-assignment can be written $\bigP_1$ or  $\bigP_2$, where $\bigP$ is some $r$-assignment. We have
\begin {align*}
S_\phi & \path \{ \SP{\bigP} \mid  \bigP\ {r}\text{-assignment} \} \text{ by induction hypothesis}\\
S_\phi & \path \{ \SP{\bigP_1},\SP{\bigP_2} \mid \bigP\ {r}\text{-assignment} \} \text{ by Property~\ref{prop:1stepSP}}\\
& = \{  \SP{\bigP'} \mid \bigP'\ {(r+1)}\text{-assignment} \}
\end {align*}
\end{proof}\fi

\subsubsection{Going through clauses}\label{red:selection}
Now that each variable is assigned a boolean value, we need to verify with each clause that this assignment satisfies the formula $\phi$. This is done by selecting, for each clause, a literal which is true, and testing the corresponding lock.
As in Definition \ref{defn:sphi}, for any $i\in\intvl1k$ we write $(a_i,b_i,c_i)$ the indices such that the $i$-th clause of $\phi$ is $\lambda_{a_i}\vee \lambda_{b_i}\vee  \lambda_{c_i}$ (thus, $a_i,b_i,c_i \in \intvl1m$).
\ifisFullVersion
\begin{defn}\label{defn:sps}
Let $t\in \intvl 0k$ and \bigP\ be a full assignment.
A \emph{$t$-selection} \sel\ is a subset of $\intvl1m$ such that
\begin{itemize}
\item $|\sel|=t$
\item for each $i\in\intvl 1t$, $|\{a_i,b_i,c_i\}\cap \sel|=1$
\end{itemize}

A $t$-selection $\sel$ and a full assignment $\bigP=(T,F)$ are \emph{compatible}, if, for every $i\in\sel$, literal $\lambda_i$ is true according to assignment $\bigP$ (that is, $\lambda_i=x_j$ and $j\in T$, or $\lambda_i=\neg x_j$ and $j\in F$).

A $k$-selection is called a \emph{full} selection.
Given a $t$-selection $\sel$ and a full assignment $\bigP=(T,F)$ which are compatible, we define the sequence $\SPs{\bigP}{\sel}$ by:
\begin{align*}
&\text{For all }i\in \intvl1l, \quad V_i' = 
\begin{cases} 
V_i^1\text{ if }i\in T\\
V_i^2\text{ if }i\in F\\
\end{cases}\\
&\text{For all }i\in \intvl1t, \quad \Gamma_i' = 
\begin{cases} 
\Gamma_i^1\text{ if }a_i\in \sel\\
\Gamma_i^2\text{ if }b_i\in \sel\\
\Gamma_i^3\text{ if }c_i\in \sel
\end{cases}\\
&O=\bigcup_{i\in T} P_i \cup \bigcup_{i\in F} N_i-\sel\\
&I=\sel \\
&\SPs{\bigP}{\sel}=
\row{
\gamma_{t+1},\ldots,\gamma_k,
V_1',\ldots,V_l',
\Gamma_1',\ldots,\Gamma_{t}',
\Gamma_{t+1},\ldots,\Gamma_k,
D_1,\ldots,D_l,
\Delta_1,\ldots,\Delta_k,
\biglambda{xx}{O}{I}
}
\end{align*}
\end{defn}

\begin{property} \label{prop:SPsred}
Let $\bigP$ be a full assignment and $t\in \intvl0{k}$, $t<k$. 
Let $\sel'$ be a $(t+1)$-selection compatible with $\bigP$, then there exists a $t$-selection $\sel$ compatible with $\bigP$ such that $\sel\subset \sel'$.
\end{property}
\begin{proof}
 It is obtained by $\sel=\sel'-\{a_{t+1},b_{t+1},c_{t+1}\}$. It is trivially  a $t$-selection included in $\sel$, and it is compatible with $\bigP$ (all selected literals in $\sel$ are also selected in $\sel'$, and thus are true according to $\bigP$).
\end{proof}
\begin{property}\label{prop:1stepSPs}
 Let $t\in\intvl0k$, $t<k$, $\bigP$ be a full assignment, and $\sel$ be a $t$-selection compatible with $\bigP$.
\begin{equation*}
\SPs{\bigP}{\sel} \path
\left\{  \SPs{\bigP}{\sel'}\mid
   \sel'\ (t+1)\text{-selection compatible with }\bigP; \sel\subset\sel'
\right\}
\end{equation*}
Note that the right-hand side can be the empty set, in which case $\SPs{\bigP}{\sel}\pathDead$.
\end{property}

\begin{proof}
First note that there are 3 $(t+1)$-selections such that $\sel\subset\sel'$, and they are 
$\sel'_1=\sel\cup\{a_{t+1}\}$, $\sel'_2=\sel\cup\{b_{t+1}\}$, and $\sel'_3=\sel\cup\{c_{t+1}\}$.
Since $\sel$ is compatible with $\bigP$, $\sel'_1$ is compatible with $\bigP$ iff literal $\lambda_{a_{t+1}}$ is true in $\bigP$ (and similarly with  couples $(\sel'_2, \lambda_{b_{t+1}})$ and $(\sel'_3, \lambda_{c_{t+1}})$). 
We now define sequences $X$ and $Y$ and sets $I$ and $O$ such that $\SPs{\bigP}{\sel}=\row{ \gamma_{t+1} \n X \n \Gamma_{t+1} \n Y  \n\biglambda{C}{O}{I} }$, that is:
\begin{align*}
&X=\row{\gamma_{t+2},\ldots,\gamma_k,V_1',\ldots,V_l',\Gamma_1',\ldots,\Gamma_{t}'} \\
&Y=\row{\Gamma_{t+2},\ldots,\Gamma_k,
D_1,\ldots,D_l,\Delta_1,\ldots,\Delta_k,}\\
&O=\bigcup_{i\in T} P_i \cup \bigcup_{i\in F} N_i-\sel\\
&I=\sel 
\end{align*}

Using Property~\ref{prop:clause} on clause gadget $(\gamma_{t+1},\Gamma_{t+1},\Delta_{t+1})$, we obtain:
\begin{equation*}
\SPs{\bigP}{\sel}\path\mathbb T
 \end{equation*}
where $\mathbb T$ is defined by:
\begin{eqnarray*}
\row{  X \n \Gamma^1_{t+1}\n Y   \n \biglambda{C}{O-\{a_{t+1}\}}{I\cup\{a_{t+1}\}} } \in \mathbb T &\mbox{ iff }& a_{t+1}\in O
\\
\row{  X \n \Gamma^2_{t+1}\n Y  \n \biglambda{C}{O-\{b_{t+1}\}}{I\cup\{b_{t+1}\}} } \in \mathbb T &\mbox{ iff }& b_{t+1}\in O
\\
\row{ X \n \Gamma^3_{t+1}\n Y   \n \biglambda{C}{O-\{c_{t+1}\}}{I\cup\{c_{t+1}\}} } \in \mathbb T &\mbox{ iff }& c_{t+1}\in O
\end{eqnarray*}

Note that $a_{t+1}\notin\sel$, hence $a_{t+1}\in O$ iff $\exists i\in T$ s.t. $a_{t+1}\in P_i$ or $\exists i\in F$ s.t. $a_{t+1}\in N_i$. Equivalently,  $a_{t+1}\in O$ iff $\lambda_{a_{t+1}}$ is a positive occurrence of a variable assigned True in $\bigP$, or a negative occurrence of a variable assigned False in $\bigP$. Finally, $a_{t+1}\in O$ iff $\sel'_1$ is compatible with $\bigP$. Likewise, $b_{t+1}\in O$ iff $\sel'_2$ is compatible with $\bigP$, and $c_{t+1}\in O$ iff $\sel'_3$ is compatible with $\bigP$.

\begin{eqnarray*}
\SPs{\bigP}{\sel'_1}=\row{  X \n \Gamma^1_{t+1}\n Y \n \biglambda{C}{O-\{a_{t+1}\}}{I\cup\{a_{t+1}\}} } \in \mathbb T &\mbox{ iff }&\sel'_1 \text{ is compatible with } \bigP 
\\
\SPs{\bigP}{\sel'_2}=\row{  X \n \Gamma^2_{t+1}\n Y \n \biglambda{C}{O-\{b_{t+1}\}}{I\cup\{b_{t+1}\}} } \in \mathbb T &\mbox{ iff }&\sel'_2 \text{ is compatible with } \bigP 
\\
\SPs{\bigP}{\sel'_3}=\row{  X \n \Gamma^3_{t+1}\n Y \n \biglambda{C}{O-\{c_{t+1}\}}{I\cup\{a_{t+1}\}} } \in \mathbb T &\mbox{ iff }&\sel'_3 \text{ is compatible with } \bigP 
\end{eqnarray*}

Thus $\mathbb T$ is indeed the set of sequences \SPs{\bigP}{\sel'} where $\sel'$ is a $(t+1)$-selection which contains $\sel$ and is compatible with $\bigP$: the property is proved. 
\end{proof}
\fi
With the following lemma, we ensure that after the truth assignment, every efficient path starting from $S_\phi$ 
needs to select a literal in each clause, under the constraint that the selection is compatible with the assignment.

\begin{lemma} \label{lem:fruClause}
Let $\bigP$ be a full assignment. Then
\begin{equation*}
 \SP{\bigP} \path \left\{ \SPs{\bigP}{\sel} \mid \sel \text{ full selection compatible with }\bigP  \right\}
\end{equation*}

\end{lemma}

\ifisFullVersion\begin{proof}
The proof follows the same pattern as the one of Lemma~\ref{lem:fruVar}, that is, 
we prove 
\begin{equation*}
 \SP{\bigP} \path \left\{ \SPs{\bigP}{\sel} \mid \sel\ t\text{-selection compatible with }\bigP  \right\}
\end{equation*}
 by induction for all $t\in\intvl 0k$, and the lemma is deduced from the case $t=k$. 

There is only one 0-selection, which is $\sel_0=\emptyset$, it is compatible with $\bigP$, and $\SP{\bigP}=\SPs{\bigP}{\sel_0}$.
Consider now any $t<k$.  We have
\begin{equation*}
\SP{\bigP}  \path \{ \SPs{\bigP}{\sel} \mid  \sel\ t\text{-selection compatible with }\bigP \} \text{ (by induction hypothesis)}
\end{equation*}
\begin{equation*}
\begin{split}
\SP{\bigP}  \path \{\SPs{\bigP}{\sel'} \mid {}&\sel'\ (t+1)\text{-selection compatible with }\bigP \text{ and }\\
 &\exists \,\sel\ t\text{-selection compatible with }\bigP,\, \sel\subset\sel'\}
 \text{ by Property~\ref{prop:1stepSPs}}
\end{split}
\end{equation*}\begin{equation*}
 = \{  \SPs{\bigP}{\sel'} \mid \sel'\ (t+1)\text{-selection compatible with }\bigP \} \text{ by Property~\ref{prop:SPsred}}
\end{equation*}
\end{proof}\fi

\subsubsection{Beyond clauses}\label{red:ending}
\begin{lemma}\label{lem:finish}
 Let $\bigP$ be a full assignment and $\sel$ be a full selection, such that $\bigP$ and $\sel$ are compatible (provided such a pair exists for $\phi$). Then
\begin{equation*}
 \SPs{\bigP}{\sel} \path \ident{1}{n}
\end{equation*}

\end{lemma}

\ifisFullVersion\begin{proof} Write $\bigP=(T,F)$.
 Since $\sel$ is a full selection,  \SPs{\bigP}{\sel} can be written (see Definition~\ref{defn:sps}):
\begin{align*}
&\text{For all }i\in \intvl1l, \quad V_i' = 
\begin{cases} 
V_i^1\text{ if }i\in T\\
V_i^2\text{ if }i\in F\\
\end{cases}\\
&\text{For all }i\in \intvl1k, \quad \Gamma_i' = 
\begin{cases} 
\Gamma_i^1\text{ if }a_i\in \sel\\
\Gamma_i^2\text{ if }b_i\in \sel\\
\Gamma_i^3\text{ if }c_i\in \sel
\end{cases}\\
&O=\bigcup_{i\in T} P_i \cup \bigcup_{i\in F} N_i-\sel\\
&I=\sel \\
&\SPs{\bigP}{\sel}=\row{V_1',\ldots,V_l',\Gamma_1',\ldots,\Gamma_{k}',
D_1,\ldots,D_l,\Delta_1,\ldots,\Delta_k,\biglambda{xx}{O}{I}
}
\end{align*}
We extend the definition of set $O$ to $O_r$, for any $r\in\intvl0l$, as follows:
\begin{equation*}
 O_r=\bigcup_{0<i\leq r} (P_i\cup N_i)\cup \bigcup_{i\in T} P_i \cup \bigcup_{i\in F} N_i-\sel\\
\end{equation*}
Note that $O_0=O$, and that $O_l=\intvl1m-\sel$.

\begin{align*}
  \SPs{\bigP}{\sel} &=\row{\ch{V_1'},\ldots,V_l',\Gamma_1',\ldots,\Gamma_{k}',
\ch{D_1},\ldots,D_l,\Delta_1,\ldots,\Delta_k,\ch{\biglambda{xx}{O_0}{I}}
 } 
\\ 
&\rp{variable}{b/c}\path
 \row{\ch{V_2'},\ldots,V_l',\Gamma_1',\ldots,\Gamma_{k}',
\ident{1}{31},\ch{D_2}\ldots,D_l,\Delta_1,\ldots,\Delta_k,\ch{\biglambda{xx}{O_1}{I}}
 } 
\\
&\rp{variable}{b/c}\path
 \row{\ch{V_3'},\ldots,V_l',\Gamma_1',\ldots,\Gamma_{k}',
\ident{1}{31},\ident{32}{62},\ch{D_3}\ldots,D_l,\Delta_1,\ldots,\Delta_k,\ch{\biglambda{xx}{O_2}{I}}
 } 
\\
&\cdots
\\
&\rp{variable}{b/c}\path
 \row{\Gamma_1',\ldots,\Gamma_{k}',
\ch{\ident{1}{31},\ident{32}{62},\ldots,\ident{31l-30}{31l}},\Delta_1,\ldots,\Delta_k,\biglambda{xx}{O_l}{I}
 } 
\\
&=
 \row{\Gamma_1',\ldots,\Gamma_{k}',
\ident{1}{31l},\Delta_1,\ldots,\Delta_k,\biglambda{xx}{O_l}{I}
 } 
\end{align*}

Finally, for the last part, we use a similar procedure, with the following sets, for $t\in\intvl0k$:
\begin{align*}
 &O'_t=\intvl1m-\left(\sel\cup\bigcup_{0<i\le t} \{a_i,b_i,c_i\}\right)\\
 &I'_t=\sel\cup\bigcup_{0<i\le t} \{a_i,b_i,c_i\}
\end{align*}

Note that $O'_0=O_l$, $I'_0=I$, $O'_k=\emptyset$, $I'_k=\intvl1m$, and more importantly, for $i>t$, assuming that $a_i\in \sel$ (cases $b_i\in \sel$ and $c_i\in\sel$ are similar), then $a_i\in I'_t$, $b_i\in O'_t$ and $c_i\in O'_t$. Hence we can successively apply Property~\ref{prop:clause2} (either .a, .b or .c) on each clause gadgets.

\begin{align*}
&\row{\ch{\Gamma_1'},\ldots,\Gamma_{k}',\ident{1}{31l},\ch{\Delta_1},\ldots,\Delta_k,\ch{\biglambda{xx}{O'_0}{I'_0}}}
 \\
&\rp{clause2}{}\path
 \row{\ch{\Gamma_2'},\ldots,\Gamma_{k}',
\ident{1}{31l},\ident{31l+1}{31l+62},\ch{\Delta_2},\ldots,\Delta_k,\ch{\biglambda{xx}{O'_1}{I'_1}}
 } 
\\
&\rp{clause2}{}\path
 \row{\ch{\Gamma_3'},\ldots,\Gamma_{k}',
\ident{1}{31l},\ident{31l+1}{31l+62},\ident{31l+63}{31l+124},\ch{\Delta_3},\ldots,\Delta_k,\ch{\biglambda{xx}{O'_2}{I'_2}}
 } 
\\
&\cdots
\\
&\rp{clause2}{}\path
 \row{\ident{1}{31l},\ch{\ident{31l+1}{31l+62},\ident{31l+63}{31l+124},\ldots,\ident{31l+62k-61}{31l+62k}},\biglambda{xx}{O'_k}{I'_k}
 } 
\\
&= \row{\ident{1}{31l},\ident{31l+1}{31l+62k},\biglambda{xx}{\emptyset}{\intvl1m}}
\\
&= \row{\ident{1}{31l},\ident{31l+1}{31l+62k},\ident{31l+62k+1}{31l+62k+12m}}
\\
&= \ident{1}{n} 
\end{align*}
\end{proof}\fi

\begin{thm}\label{thm:big}
 \begin{equation*}
  S_\phi \path \ident{1}{n} \mbox{ iff } \phi \mbox{ is satisfiable.}
 \end{equation*}
\end{thm}

\begin{proof}
 Assume first that $S_\phi \path \ident{1}{n}$.
 By Lemma~\ref{lem:fruVar}, since $S_\phi\path  \left\{ \SP{\bigP} \mid \bigP \text{ full assignment} \right\} $, there exists a full assignment $\bigP=(T,F)$ such that the path from $S_\phi$ to the identity uses $\SP{\bigP}$. Note that  $\SP{\bigP} \path \ident{1}{n}$.
Now, by Lemma~\ref{lem:fruClause}, since $ \SP{\bigP} \path \left\{ \SPs{\bigP}{\sel} \mid \sel \text{ full selection compatible with }\bigP  \right\}$, there exists a full selection $\sel$, compatible with $\bigP$, such that the path from $\SP{\bigP}$ to the identity uses \SPs{\bigP}{\sel}. 
Consider the truth assignment $x_i:=$ True $\Leftrightarrow i\in T$. Then each clause of $\phi$ contains at least one literal that is true (the literal whose index is in $\sel$), and thus $\phi$ is satisfiable.

Assume now that $\phi$ is satisfiable: consider any truth assignment making $\phi$ true, write $T$ the set of indices such that $x_i=$ True, and  $F=\intvl1l-T$. Write also $\sel$ a set containing, for each clause of $\phi$, the index of one literal being true under this assignment. Then $\sel$ is a full selection, compatible with the full assignment $\bigP=(T,F)$.
By Lemma~\ref{lem:fruVar}, there exists an efficient path from $S_\phi$ to $\SP{\bigP}$. By Lemma~\ref{lem:fruClause}, there exists an efficient path from $\SP{\bigP}$ to $\SPs{\bigP}{\sel}$. And by Lemma~\ref{lem:finish}, there exists an efficient path from  $\SPs{\bigP}{\sel}$ to the identity. Thus sequence $S_\phi$ is efficiently sortable.
\end{proof}

Using Theorem~\ref{thm:big}, we can now prove the main result of the paper.
\begin{thm}
The following problems are \textsf{NP}-hard:
\begin{itemize}
 \item Sorting By Prefix Reversals (MIN-SBPR) 
 \item deciding, given a sequence $S$, whether $S$ can be sorted in $d_b(S)$ flips
\end{itemize}
\end{thm}
\begin{proof}
 By reduction from 3-SAT. Given any formula $\phi$, create $S_\phi$ (see Definition~\ref{defn:sphi}, the construction requires a linear time). By Theorem~\ref{thm:big}, the minimum number of flips necessary to sort $S_\phi$ is $d_b(S_\phi)$ iff $\phi$ is satisfiable.
\end{proof}

\section{Conclusion}

In this paper, we have shown that the Pancake Flipping problem is \textsf{NP}-hard, thus answering a long-standing open question.
We have also provided a stronger result, namely, deciding whether a permutation can be sorted with no more than one flip per breakpoint is also \textsf{NP}-hard.

Among related important problems, the last one having an open complexity is now the burnt variant of the Pancake Flipping problem.
An interesting insight into this problem is given in a recent work from Labarre and Cibulka~\cite{LC11}, where the authors characterize a subclass of permutations that can be sorted in polynomial time, using the breakpoint graph~\cite{BP93}.
Another development consists in trying to improve the approximation ratio of 2 for the Pancake Flipping problem, both in its burnt and unburnt versions. 
\newpage
\bibliographystyle{plain}
\bibliography{biblio}

\end{document}